\documentclass{aims}
\usepackage{amsmath}
  \usepackage{paralist}
  \usepackage{graphics} 
  \usepackage{epsfig} 
\usepackage{graphicx}  \usepackage{epstopdf}
 \usepackage[colorlinks=true]{hyperref}
\hypersetup{urlcolor=blue, citecolor=red}

\def\currentvolume{X}
 \def\currentissue{X}
  \def\currentyear{200X}
   \def\currentmonth{XX}
    \def\ppages{X--XX}
     \def\DOI{10.3934/xx.xx.xx.xx}

\newtheorem{theorem}{Theorem}[section]
\newtheorem{corollary}[theorem]{Corollary}

\newtheorem{lemma}[theorem]{Lemma}

\theoremstyle{definition}

\newtheorem*{remark}{Remark}
\newtheorem{notation}[theorem]{Notation}
\newtheorem{example}[theorem]{Example}

\title[Constacyclic codes
of length $np^s$ over $\mathbb{F}_{p^m}+u\mathbb{F}_{p^m}$] 
      {Constacyclic codes
of length $np^s$ over $\mathbb{F}_{p^m}+u\mathbb{F}_{p^m}$}

\author[Yonglin Cao, Yuan Cao, Jian Gao  and Fang-Wei Fu]{}


 \email{ylcao@sdut.edu.cn}
 \email{yuancao@sdut.edu.cn}
  \email{dezhougaojian@163.com}
 \email{fwfu@nankai.edu.cn}

\thanks{ This research is supported in part by National Natural Science Foundation of
China (Nos. 11671235, 11626144,
11471255) and the National Key Basic Research Program of China (Grant 2013CB834204).}

\thanks{$^*$ Corresponding author: Yuan Cao}

\begin{document}
\maketitle

\centerline{\scshape Yonglin Cao, Yuan Cao$^*$, Jian Gao}
\medskip
{\footnotesize
 \centerline{School of Mathematics and  Statistics, Shandong University of Technology, }
   \centerline{ Zibo, Shandong 255091, China}
} 
\medskip

\centerline{\scshape Fang-Wei Fu}
\medskip
{\footnotesize
 \centerline{ Chern Institute of Mathematics and LPMC, Nankai University}
   \centerline{Tianjin 300071, China}
}
\bigskip

 \centerline{(Communicated by the associate editor name)}

\begin{abstract}
Let $\mathbb{F}_{p^m}$ be a finite field of cardinality $p^m$ and $R=\mathbb{F}_{p^m}[u]/\langle u^2\rangle=\mathbb{F}_{p^m}+u\mathbb{F}_{p^m}$ $(u^2=0)$, where $p$ is a prime and $m$ is a positive integer. For any $\lambda\in \mathbb{F}_{p^m}^{\times}$,
an explicit representation for all distinct $\lambda$-constacyclic codes
over $R$ of length $np^s$ is given by a canonical form decomposition for each code, where $s$ and $n$ are
arbitrary positive integers satisfying ${\rm gcd}(p,n)=1$. For any such code, using its canonical form decomposition the representation for the dual code of the code is provided. Moreover, representations for all distinct cyclic codes,
negacyclic codes and their dual codes of length $np^s$ over $R$ are obtained,
and self-duality for these codes are determined. Finally, all distinct self-dual negacyclic codes over $\mathbb{F}_5+u\mathbb{F}_5$ of
length $2\cdot 3^t\cdot 5^s$ are listed for any positive integer $t$.
\end{abstract}

\section{Introduction}

\noindent \par
Algebraic coding theory deals with the design of error-correcting and error-detecting codes for the reliable transmission
of information across noisy channel. The class of constacyclic codes plays a very significant role in
the theory of error-correcting codes as they can be efficiently encoded with simple shift
registers. This family of codes is thus interesting for both theoretical and practical reasons.

\par
  Let $\Gamma$ be a commutative finite ring with identity $1\neq 0$, and $\Gamma^{\times}$ be the multiplicative group of invertible elements of
$\Gamma$. For any $a\in
\Gamma$, we denote by $\langle a\rangle_\Gamma$, or $\langle a\rangle$ for
simplicity, the ideal of $\Gamma$ generated by $a$, i.e., $\langle
a\rangle_\Gamma=a\Gamma=\{ab\mid b\in \Gamma\}$. For any ideal $I$ of $\Gamma$, we will identify the
element $a+I$ of the residue class ring $\Gamma/I$ with $a$ (mod $I$) for
any $a\in \Gamma$ in this paper.

\par
   A \textit{code} over $\Gamma$ of length $N$ is a nonempty subset ${\mathcal C}$ of $\Gamma^N=\{(a_0,a_1,\ldots$, $a_{N-1})\mid a_j\in\Gamma, \
j=0,1,\ldots,N-1\}$. The code ${\mathcal C}$
is said to be \textit{linear} if ${\mathcal C}$ is a $\Gamma$-submodule of $\Gamma^N$. All codes in this paper are assumed to be linear. The ambient space $\Gamma^N$ is equipped with the usual Euclidian inner product, i.e.
$[a,b]=\sum_{j=0}^{N-1}a_jb_j$, where $a=(a_0,a_1,\ldots,a_{N-1}), b=(b_0,b_1,\ldots,b_{N-1})\in \Gamma^N$.
Then the \textit{dual code} of ${\mathcal C}$ is defined by ${\mathcal C}^{\bot}=\{a\in \Gamma^N\mid [a,b]=0, \forall b\in {\mathcal C}\}$.
If ${\mathcal C}^{\bot}={\mathcal C}$, ${\mathcal C}$ is called a \textit{self-dual code} over $\Gamma$.

\par
   Let $\gamma\in \Gamma^{\times}$.
Then a linear code
${\mathcal C}$ over $\Gamma$ of length $N$ is
called a $\gamma$-\textit{constacyclic code}
if $(\gamma c_{N-1},c_0,c_1,\ldots,c_{N-2})\in {\mathcal C}$ for all
$(c_0,c_1,\ldots,c_{N-1})\in{\mathcal C}$. Particularly, ${\mathcal C}$ is
called a \textit{negacyclic code} if $\gamma=-1$, and ${\mathcal C}$ is
called a  \textit{cyclic code} if $\gamma=1$.
  For any $a=(a_0,a_1,\ldots,a_{N-1})\in \Gamma^N$, let
$a(x)=a_0+a_1x+\ldots+a_{N-1}x^{N-1}\in \Gamma[x]/\langle x^N-\gamma\rangle$. We will identify $a$ with $a(x)$ in
this paper. By \cite{s9} Propositions 2.2 and 2.4, we have

\begin{lemma}
Let $\gamma\in \Gamma^{\times}$. Then ${\mathcal C}$ is a  $\gamma$-constacyclic code
over $\Gamma$ of length $N$ if and only if ${\mathcal C}$ is an ideal of
the residue class ring $\Gamma[x]/\langle x^N-\gamma\rangle$.
\end{lemma}

\begin{lemma}
The dual code of a $\gamma$-constacyclic code over
$\Gamma$ of length $N$ is a $\gamma^{-1}$-constacyclic code over
$\Gamma$ of length $N$, i.e., an ideal of $\Gamma[x]/\langle
x^N-\gamma^{-1}\rangle$.
\end{lemma}

  In this paper, let $\mathbb{F}_{p^m}$ be a finite field of cardinality $p^m$, where
$p$ is a prime and $m$ is a positive integer, and denote $\mathbb{F}_{p^m}[u]/\langle u^2\rangle$
by $\mathbb{F}_{p^m}+u\mathbb{F}_{p^m}$ ($u^2=0$). There were a lot of literatures on linear codes, cyclic codes and
constacyclic codes of length $N$ over rings $\mathbb{F}_{p^m}+u\mathbb{F}_{p^m}$ ($u^2=0$) for various prime $p$ and positive integers $m$ and some positive integer $N$.
For example, \cite{s1}, \cite{s2}, \cite{s4}, \cite{s7}--\cite{s14}, \cite{s16}, \cite{s17} and \cite{s20}.
The classification of codes plays an important role in studying their structures and encoders.
However, it is a very difficult task in general, and only some codes of special lengths over certain finite
fields or finite chain rings are classified.

\par
  For example,
all constacyclic codes of length $2^s$ over the Galois extension
rings of $\mathbb{F}_2 + u\mathbb{F}_2$ was classified and their detailed structures was also established in \cite{s8}. Dinh \cite{s9}
classified all constacyclic codes of length $p^s$ over $\mathbb{F}_{p^m}+u\mathbb{F}_{p^m}$.
    Dinh et al. \cite{s10} studied
negacyclic codes of length $2p^s$ over the ring $\mathbb{F}_{p^m}+u\mathbb{F}_{p^m}$.
Chen et al. \cite{s7} investigated
constacyclic codes of length $2p^s$ over $\mathbb{F}_{p^m}+u\mathbb{F}_{p^m}$.
Recently, Dinh et al. \cite{s11} \cite{s12} studied constacyclic codes of length $4p^s$ over $\mathbb{F}_{p^m}+u\mathbb{F}_{p^m}$.
These papers mainly used the methods in \cite{s8} and \cite{s9}, and the main results and their proofs  depend heavily on the code lengths $p^s$, $2p^s$ and $4p^s$.

\par
   From now on, let $n,s$  be arbitrary positive integers satisfying ${\rm gcd}(p,n)=1$,
and $\lambda$ be an arbitrary nonzero element of $\mathbb{F}_{p^m}$. In this paper, by use of
basic theory for linear codes over finite chain rings of length $2$, we provide a new way different from the methods used in \cite{s7}--\cite{s13} to determine  precisely the algebraic structures, the
generators and enumeration of $\lambda$-constacyclic codes over $\mathbb{F}_{p^m}+u\mathbb{F}_{p^m}$ of length $np^s$. Specifically, we will address the following questions:

\par
  $\diamondsuit$ Give a precise representation for each $\lambda$-constacyclic code $\mathcal{C}$ over $\mathbb{F}_{p^m}+u\mathbb{F}_{p^m}$ of length $np^s$,
and provide a simple and clear formula to count the number of codewords in $\mathcal{C}$.

\par
  $\diamondsuit$ Give a clear  formula to count the number of all $\lambda$-constacyclic codes over $\mathbb{F}_{p^m}+u\mathbb{F}_{p^m}$ of length $np^s$.

\par
  $\diamondsuit$ Give a precise representation for the dual code of each $\lambda$-constacyclic code $\mathcal{C}$ over $\mathbb{F}_{p^m}+u\mathbb{F}_{p^m}$ of length $np^s$ by use of the representation of $\mathcal{C}$.

\par
  $\diamondsuit$ Determine self-dual cyclic and negacyclic codes over $\mathbb{F}_{p^m}+u\mathbb{F}_{p^m}$ of length $np^s$.

\begin{notation} In the rest of the paper, we denote
        \vskip 2mm \par
  $\bullet$ $R=\mathbb{F}_{p^m}[u]/\langle u^2\rangle=\mathbb{F}_{p^m}
+u\mathbb{F}_{p^m}$ ($u^2=0$);

\vskip 2mm \par
  $\bullet$ $\mathcal{A}=\mathbb{F}_{p^m}[x]/\langle x^{np^s}-\lambda\rangle$ and $\mathcal{A}[u]/\langle u^2\rangle=\mathcal{A}+u\mathcal{A}$ ($u^2=0$);

\vskip 2mm \par
   $\bullet$ $\mathcal{R}_{\lambda}=R[x]/\langle x^{np^s}-\lambda\rangle$
and $\mathcal{R}_{\lambda^{-1}}=R[x]/\langle x^{np^s}-\lambda^{-1}\rangle$.
\end{notation}

\vskip 3mm\par
  The present paper is organized as follows. In Section 2, we provide a direct sum decomposition
for any $\lambda$-constacyclic code  over $R$ of length $np^s$. Then we determine each
direct summand of such decomposition in Section 3. Hence we obtain an explicit representation for each of these codes and give a formula to count the number of codewords in each code from its representation. As a corollary, we obtain a formula to enumerate all such codes.
Then we determine the dual code of each code in Section 4. In Section 5, we list all distinct cyclic codes,
negacyclic codes and their dual codes  over $R$ of length $np^s$, and determine the self-duality for such codes.
Particularly, we present all distinct self-dual negacyclic codes over $\mathbb{F}_5+u\mathbb{F}_5$ of
length $2\cdot 3^t\cdot 5^s$ for any positive integer $t$ in Section 6. Finally, in Section 7 we obtain conclusions for the special cases of $n=1,2,4$ which are match known results in \cite{s7}--\cite{s12}.




\section{Decomposition for $\lambda$-constacyclic codes over $R$ of length $np^s$}

\noindent \par
In this section, we construct a specific ring isomorphism from $\mathcal{A}+u\mathcal{A}$ ($u^2=0$) onto
$\mathcal{R}_\lambda$. By use of this isomorphism, we obtain a one-to-one correspondence
between the set of ideals of $\mathcal{A}+u\mathcal{A}$ and the set of ideals of $\mathcal{R}_\lambda$, i.e.,
the set of $\lambda$-constacyclic codes over $R$ of length $np^s$.

\par
   Let $g(x)\in \mathcal{R}_\lambda$. Then $g(x)$ can be uniquely
expressed as
$$g(x)=\sum_{j=0}^{np^s-1}g_jx^j, \ g_j\in R \
{\rm for} \ j=0,1,\ldots,np^s-1,$$
where each element $g_j$ of $R=\mathbb{F}_{p^m}
+u\mathbb{F}_{p^m}$ is uniquely expressed as
$$g_j=g_{j,0}+ug_{j,1}, \ g_{j,0},g_{j,1}\in \mathbb{F}_{p^m}, \
j=0,1,\ldots,np^s-1.$$
Hence $g(x)=g_0(x)+ug_1(x)$ where
$g_i(x)=\sum_{j=0}^{np^s-1}g_{j,i}x^j\in \mathcal{A}$ for all $i=0,1$.
Moreover, we have the following lemma.

\begin{lemma}\label{le2.1} For any $\xi=a(x)+ub(x)\in \mathcal{A}+u\mathcal{A}$, where $a(x)=\sum_{j=0}^{np^s-1}a_jx^j$ and
$b(x)=\sum_{j=0}^{np^s-1}b_jx^j$ with $a_j,b_j\in \mathbb{F}_{p^m}$, we define
\begin{center}
$\Psi(\xi)=a(x)+ub(x)=\sum_{j=0}^{np^s-1}(a_j+ub_j)x^j.$
\end{center}
Then $\Psi$ is a ring isomorphism from $\mathcal{A}+u\mathcal{A}$ onto $\mathcal{R}_\lambda$.
\end{lemma}

\begin{proof}
It is clear that  $\Psi$ is bijection from $\mathcal{A}+u\mathcal{A}$ onto $\mathcal{R}_\lambda$.
Then by trivial calculations one can verify that $\Psi(\xi+\eta)=\Psi(\xi)+\Psi(\eta)$ and $\Psi(\xi\cdot\eta)=\Psi(\xi)\cdot\Psi(\eta)$
for any $\xi,\eta\in \mathcal{A}+u\mathcal{A}$.
\end{proof}

   In the rest of this paper, we will identify $\mathcal{A}+u\mathcal{A}$ with $\mathcal{R}_\lambda$
under the ring isomorphism $\Psi$ defined in Lemma \ref{le2.1}. Then
in order to determine all distinct $\lambda$-constacyclic codes over $R$ of length $np^s$,
it is sufficient to list all distinct
ideals of the ring $\mathcal{A}+u\mathcal{A}$. To do this, we need to investigate the structure of the
ring $\mathcal{A}=\mathbb{F}_{p^m}[x]/\langle x^{p^sn}-\lambda\rangle$ first.

\par
   Since $\lambda\in \mathbb{F}_{p^m}^{\times}$ and
$\mathbb{F}_{p^m}^{\times}$ is a multiplicative cyclic group of order $p^m-1$, there is a unique element
$\lambda_0\in \mathbb{F}_{p^m}^{\times}$ such that $\lambda_0^{p^s}=\lambda$. From this, we deduce that
$x^{np^s}-\lambda=(x^n-\lambda_0)^{p^s}$ in $\mathbb{F}_{p^m}[x]$. As ${\rm gcd}(p,n)=1$, there are pairwise coprime monic
irreducible polynomials $f_1(x),\ldots, f_r(x)$ in $\mathbb{F}_{p^m}[x]$ such that
$x^n-\lambda_0=f_1(x)\ldots f_r(x)$. This implies
\begin{equation}\label{eq1}x^{np^s}-\lambda=(x^n-\lambda_0)^{p^s}=f_1(x)^{p^s}\ldots f_r(x)^{p^s}.
\end{equation}
For any integer $j$, $1\leq j\leq r$, we assume ${\rm deg}(f_j(x))=d_j$ and denote $F_j(x)=\frac{x^{n}-\lambda_0}{f_j(x)}$.
Then $F_j(x)^{p^s}=\frac{x^{np^s}-\lambda}{f_j(x)^{p^s}}$ and ${\rm gcd}(F_j(x),f_j(x))=1$. Hence there exist $v_j(x),w_j(x)\in \mathbb{F}_{p^m}[x]$ such that
${\rm deg}(v_j(x))<{\rm deg}(f_j(x))=d_j$ and $v_j(x)F_j(x)+w_j(x)f_j(x)=1$. This implies
\begin{equation}\label{eq2}v_j(x)^{p^s}F_j(x)^{p^s}+w_j(x)^{p^s}f_j(x)^{p^s}
=(v_j(x)F_j(x)+w_j(x)f_j(x))^{p^s}=1.
\end{equation}

\par
  In the rest of this paper, we adopt the following notations.
\begin{notation}
\label{no2.2}Let $1\leq j\leq r$. Using the notations above, we denote
\begin{center}
$\mathcal{K}_j=\mathbb{F}_{p^m}[x]/\langle f_j(x)^{p^s}\rangle$, $\mathcal{K}_j[u]/\langle u^2\rangle=\mathcal{K}_j+u\mathcal{K}_j$
$(u^2=0)$
\end{center}
and let $\varepsilon_j(x)\in\mathcal{A}$ satisfying
\begin{equation}\label{eq3}
\varepsilon_j(x)= v_j(x)^{p^s}F_j(x)^{p^s}=1-w_j(x)^{p^s}f_j(x)^{p^s} \ ({\rm mod} \ x^{np^s}-\lambda).
\end{equation}
\end{notation}

\vskip 3mm \par
  Then by Equations (\ref{eq1})-(\ref{eq3}), we deduce the following conclusions.

\begin{lemma}\label{le2.3}
Using the notations above, we have the following:
\par
  (i) $\varepsilon_1(x)+\ldots+\varepsilon_r(x)=1$, $\varepsilon_j(x)^2=\varepsilon_j(x)$
and $\varepsilon_j(x)\varepsilon_l(x)=0$  in the ring $\mathcal{A}$ for all $1\leq j\neq l\leq r$.
\par
  (ii) $\mathcal{A}=\mathcal{A}_1\oplus\ldots \oplus\mathcal{A}_r$ where $\mathcal{A}_j=\mathcal{A}\varepsilon_j(x)$ with
$\varepsilon_j(x)$ as its multiplicative identity and satisfies $\mathcal{A}_j\mathcal{A}_l=\{0\}$ for all $1\leq j\neq l\leq r$.
\par
  (iii) For any integer $j$, $1\leq j\leq r$, for any $a(x)\in \mathcal{K}_j$ we define
  \begin{center}
$\varphi_j: a(x)\mapsto \varepsilon_j(x)a(x)$ $(${\rm mod} $x^{np^s}-\lambda)$.
\end{center}
Then $\varphi_j$ is a ring isomorphism from $\mathcal{K}_j$ onto $\mathcal{A}_j$.
\par
  (iv) For any $a_j(x)\in \mathcal{K}_j$ for $j=1,\ldots,r$, define
\begin{center}
$\varphi(a_1(x),\ldots,a_r(x))=\sum_{j=1}^r\varphi_j(a_j(x))=\sum_{j=1}^r\varepsilon_j(x)a_j(x)$ $(${\rm mod} $x^{np^s}-\lambda)$.
\end{center}
Then
$\varphi$ is a ring isomorphism from $\mathcal{K}_1\times\ldots\times\mathcal{K}_r$ onto $\mathcal{A}$.
\end{lemma}

\begin{proof} (i) By Equation (\ref{eq3}) and $f_j(x)^{p^s}F_j(x)^{p^s}=x^{np^s}-\lambda$, we have $\varepsilon_j(x)^2\equiv (v_j(x)^{p^s}F_j(x)^{p^s})(1-w_j(x)^{p^s}f_j(x)^{p^s})\equiv v_j(x)^{p^s}F_j(x)^{p^s}\equiv \varepsilon_j(x)$ (mod $x^{np^s}-\lambda$) and
$\varepsilon_j(x)\varepsilon_l(x)\equiv (v_j(x)^{p^s}F_j(x)^{p^s})(v_l(x)^{p^s}F_l(x)^{p^s})\equiv 0$
(mod $x^{np^s}-\lambda$). Hence $\varepsilon_j(x)^2=\varepsilon_j(x)$ and $\varepsilon_j(x)\varepsilon_l(x)=0$
in the ring $\mathcal{A}=\mathbb{F}_{p^m}[x]/\langle x^{np^s}-\lambda\rangle$ for
all $1\leq j\neq l\leq r$. From these we deduce that $\prod_{j=1}^r(1-\varepsilon_j(x))
=1-\sum_{j=1}^r\varepsilon_j(x)$ in $\mathcal{A}$. On the other hand, we have $\prod_{j=1}^r(1-\varepsilon_j(x))
\equiv\prod_{j=1}^rw_j(x)^{p^s}f_j(x)^{p^s}\equiv 0$ (mod $x^{np^s}-\lambda$). Therefore,
we have $\sum_{j=1}^r\varepsilon_j(x)=1$ in $\mathcal{A}$.

\par
  (ii) By classical ring theory, it follows from (i) immediately.

\par
  (iii) By $\varepsilon_j(x)^2=\varepsilon_j(x)$ in $\mathcal{A}$, we see that
$\varphi_j$ is a ring homomorphism from $\mathcal{K}_j$ to $\mathcal{A}_j$. Let
$b(x)\in \mathcal{A}_j$. Dividing $b(x)$ by $f_j(x)^{p^s}$, we obtain $b(x)=c(x)f_j(x)^{p^s}+a(x)$
where $a(x),c(x)\in \mathbb{F}_{p^m}[x]$ satisfying $a(x)=0$ or ${\rm deg}(a(x))<{\rm deg}(f_j(x)^{p^s})$.
Then $a(x)\in \mathcal{K}_j=\mathbb{F}_{p^m}[x]/\langle f_j(x)^{p^s}\rangle$, and by (ii) we have
$b(x)\equiv \varepsilon_j(x)b(x)\equiv v_j(x)^{p^s}F_j(x)^{p^s}(c(x)f_j(x)^{p^s}+a(x))
\equiv v_j(x)^{p^s}F_j(x)^{p^s}a(x)\equiv \varepsilon_j(x)a(x)$ (mod $x^{np^s}-\lambda$). This
implies $b(x)=\varepsilon_j(x)a(x)$ in $\mathcal{A}_j$. Hence $\varphi_j$ is surjective.

\par
  Let $a(x)\in \mathcal{K}_j$ satisfying $\varepsilon_j(x)a(x)=0$ in $\mathcal{A}_j$. Then
$a(x)-a(x)w_j(x)^{p^s}f_j(x)^{p^s}=(1-w_j(x)^{p^s}f_j(x)^{p^s})a(x)\equiv \varepsilon_j(x)a(x)\equiv0$
(mod $x^{np^s}-\lambda$). This implies that $a(x)\equiv a(x)-a(x)w_j(x)^{p^s}f_j(x)^{p^s}\equiv 0$ (mod $f_j(x)^{p^s}$), and hence $a(x)=0$ in $\mathcal{K}_j$. Therefore, $\varphi_j$ is a ring isomorphism from $\mathcal{K}_j$ onto $\mathcal{A}_j$.

\par
 (iv) By the definition of direct product rings, it follows from (ii) and (iii).
\end{proof}

Next, we investigate the structure of the ring $\mathcal{A}+u\mathcal{A}$.
\begin{lemma}\label{le2.4}
Using the notations in Lemma \ref{le2.3}, for any $\xi_j+u\eta_j\in \mathcal{K}_j+u\mathcal{K}_j$ with $\xi_j,\eta_j\in \mathcal{K}_j$, where $j=1,\ldots,r$, we define
\begin{equation}\label{eq4}\Phi(\xi_1+u\eta_1,\ldots,\xi_r+u\eta_r)=\sum_{j=1}^r\left(\varphi_j(\xi_j)+u\varphi_j(\eta_j)\right)
=\sum_{j=1}^r\varepsilon_j(x)(\xi_j+u\eta_j).
\end{equation}
Then
$\Phi$ is a ring isomorphism from $(\mathcal{K}_1+u\mathcal{K}_1)\times\ldots\times(\mathcal{K}_r+u\mathcal{K}_r)$ onto $\mathcal{A}+u\mathcal{A}$.
\end{lemma}

\begin{proof}
 By Lemma \ref{le2.3}(iv), it is clear that the ring isomorphism $\varphi:\mathcal{K}_1\times\ldots\times\mathcal{K}_r\rightarrow\mathcal{A}$
can be extended to a polynomial ring isomorphism $\Phi_0$ from $\mathcal{K}_1[u]\times\ldots\times\mathcal{K}_r[u]$ onto $\mathcal{A}[u]$ by
the rule that
$$\Phi_0\left(\sum_t\xi_{1,t}u^t,\ldots,\sum_t\xi_{r,t}u^t\right)=\sum_t\left(\sum_{j=1}^r\varphi_j(\xi_{j,t})\right)u^t
\ (\forall \xi_{j,t}\in \mathcal{K}_j).$$
Then by classical ring theory, we see that $\Phi_0$ induces a ring isomorphism $\Phi$ from
 $\left(\mathcal{K}_1[u]/\langle u^2\rangle\right)\times\ldots\times\left(\mathcal{K}_r[u]/\langle u^2\rangle\right)$
onto $\mathcal{A}[u]/\langle u^2\rangle$ define by (\ref{eq4}). Now, the conclusion follows from
Notation \ref{no2.2} and $\mathcal{A}[u]/\langle u^2\rangle=\mathcal{A}+u\mathcal{A}$ $(u^2=0)$, immediately.
\end{proof}
  Finally, we give a direct sum decomposition for any
$\lambda$-constacyclic code over $R$ of length $np^s$.

\begin{theorem}\label{th2.5}
  Using the notations above, let $\mathcal{C}\subseteq \mathcal{R}_\lambda=R[x]/\langle x^{np^s}-\lambda\rangle$. Then the following statements are equivalent:
\vskip 2mm\par
  (i) \textit{$\mathcal{C}$ is a $\lambda$-constacyclic code over $R$ of length $np^s$,
i.e. an ideal of $\mathcal{R}_\lambda$};

\vskip 2mm \par
  (ii) \textit{$\mathcal{C}$ is an ideal of $\mathcal{A}+u\mathcal{A}$};

\vskip 2mm \par
  (iii) \textit{For each integer $j$, $1\leq j\leq r$, there is a unique ideal $C_j$ of $\mathcal{K}_j+u\mathcal{K}_j$
such that $\mathcal{C}=\bigoplus_{j=1}^r\varepsilon_j(x)C_j$ $({\rm mod} \ x^{np^s}-\lambda)$}.
\end{theorem}

\begin{proof}
(i)$\Leftrightarrow$(ii) It follows from that $\mathcal{A}+u\mathcal{A}=\mathcal{R}_\lambda$ under a ring isomorphism.

\par
  (ii)$\Leftrightarrow$(iii) By Lemma \ref{le2.1} we know that $\mathcal{C}$ is an
ideal of $\mathcal{A}+u\mathcal{A}$ if and only if there is a unique ideal $I$ of the ring
$(\mathcal{K}_1+u\mathcal{K}_1)\times\ldots\times(\mathcal{K}_r+u\mathcal{K}_r)$ such that
$\Phi(I)=\mathcal{C}$. Furthermore, by classical ring theory we see that $I$ is an ideal of
$(\mathcal{K}_1+u\mathcal{K}_1)\times\ldots\times(\mathcal{K}_r+u\mathcal{K}_r)$ if and only if
for each integer $j$, $1\leq j\leq r$, there is a unique ideal $C_j$ of $\mathcal{K}_j+u\mathcal{K}_j$
such that
\begin{center}
$I=C_1\times\ldots\times C_r=\{(\alpha_1,\ldots,\alpha_r)\mid \alpha_j\in C_j, \ j=1,\ldots,r\}$.
\end{center}
When this condition is satisfied, by Equation (\ref{eq4}) we have
\begin{eqnarray*}
\mathcal{C}&=&\Phi(I)=\{\Phi(\alpha_1,\ldots,\alpha_r)\mid \alpha_j\in C_j, \ j=1,\ldots,r\}\\
 &=&\{\sum_{j=1}^r\varepsilon_j(x)\alpha_j\mid \alpha_j\in C_j, \ j=1,\ldots,r\}.
\end{eqnarray*}
Hence $\mathcal{C}=\bigoplus_{j=1}^r\varepsilon_j(x)C_j$, since $\varepsilon_j(x)C_j=\{\varepsilon_j(x)\alpha_j\mid \alpha_j\in C_j\}$ for all $j$.
\end{proof}
  Therefore, in order to determine all distinct $\lambda$-constacyclic codes over $R$ of length $np^s$,
by Theorem \ref{th2.5} it is sufficient to list all distinct
ideals of the ring $\mathcal{K}_j+u\mathcal{K}_j$ ($u^2=0$) for all $j=1,\ldots,r$.



\section{Representation for ideals of $\mathcal{K}_j+u\mathcal{K}_j$}

\noindent \par
In this section, we determine all distinct ideals of $\mathcal{K}_j+u\mathcal{K}_j$ for all $j=1,\ldots, r$. As $\mathcal{K}_j=\mathbb{F}_{p^m}[x]/\langle f_j(x)^{p^s}\rangle$ and $f_j(x)$ is a monic irreducible polynomial in $\mathbb{F}_{p^m}[x]$ of degree $d_j$, we have the following conclusions.

\begin{lemma}\label{le3.1}
  (cf. \cite{s5} Lemma 3.7 and \cite{s6} Example 2.1) $\mathcal{K}_j$ have the following properties:

\vskip 2mm\par
  (i) $\mathcal{K}_j$ is a finite chain ring, $f_j(x)$ generates the unique
maximal ideal $\langle f_j(x)\rangle$ of $\mathcal{K}_j$, $p^s$ is the nilpotency index of $f_j(x)$ and the residue class field of $\mathcal{K}_j$ modulo $\langle f_j(x)\rangle$ is $\mathcal{K}_j/\langle f_j(x)\rangle\cong \mathbb{F}_{p^m}[x]/\langle f_j(x)\rangle$, where $\mathbb{F}_{p^m}[x]/\langle f_j(x)\rangle$ is an extension field of $\mathbb{F}_{p^m}$ with $p^{md_j}$ elements.

\vskip 2mm\par
  (ii) Let ${\mathcal T}_j=\{\sum_{i=0}^{d_j-1}t_ix^i\mid t_0,t_1,\ldots,t_{d_j-1}\in \mathbb{F}_{p^m}\}$. Then $|{\mathcal T}_j|=p^{md_j}$, and every element $\xi$ of $\mathcal{K}_j$ has a unique $f_j(x)$-adic expansion:
\begin{center}
$\xi=\sum_{k=0}^{p^s-1}b_k(x)f(x)^k$,
where $b_k(x)\in {\mathcal T}_j$ for all $k=0,1,\ldots,p^s-1$.
\end{center}

\noindent
If $\xi\neq 0$, the \textit{$f_j(x)$-degree} of $\xi$ is defined as the least index $k$ for which $b_k(x)\neq 0$ and denoted as  $\|\xi\|_{f_j(x)}=k$. If $\xi=0$
we write $\|\xi\|_{f_j(x)}=p^s$. Moreover, $\xi\in \mathcal{K}_j^{\times}$ if and only if
$b_0(x)\neq 0$, i.e. $\|\xi\|_{f_j(x)}=0$.

\vskip 2mm\par
  (iii) All distinct ideals of $\mathcal{K}_j$ are given by: $\langle f_j(x)^l\rangle=f_j(x)^l\mathcal{K}_j$, $l=0,1,\ldots,p^s$. Moreover, $|\langle f_j(x)^l\rangle|=p^{md_j(p^s-l)}$ for $l=0,1,\ldots,p^s$.

\vskip 2mm\par
  (iv) Let $1\leq l\leq p^s$. Then
$\mathcal{K}_j/\langle f_j(x)^l\rangle=\{\sum_{k=0}^{l-1}b_k(x)f(x)^k\mid b_k(x)\in \mathcal{T}_j, \ k=0,1,\ldots,l-1\}$ and
hence $|\mathcal{K}_j/\langle f_j(x)^l\rangle|=p^{md_jl}$.

\vskip 2mm\par
  (v) For any $0\leq l\leq t\leq p^s-1$, we have
\begin{center}
$f_j(x)^l(\mathcal{K}_j/\langle f_j(x)^t\rangle)=\{\sum_{k=l}^{t-1}b_k(x)f(x)^k\mid
b_k(x)\in \mathcal{T}_j, \ k=l,\ldots,t-1\}$
\end{center}
and $|f_j(x)^l(\mathcal{K}_j/\langle f_j(x)^t\rangle)|=p^{md_j(t-l)}$, where we set
$f_j(x)^l(\mathcal{K}_j/\langle f_j(x)^l\rangle)=\{0\}$ for convenience.
\end{lemma}

   By Notation \ref{no2.2}, the addition and multiplication on the ring $\mathcal{K}_j+u\mathcal{K}_j$ are defined by:
for any $\xi_0,\xi_1,\eta_0,\eta_1\in \mathcal{K}_j$,

\vskip 2mm\par
   $\diamondsuit$ $(\xi_0+u\xi_1)+(\eta_0+u\eta_1)=(\xi_0+\eta_0)+u(\xi_1+\eta_1)$;

\vskip 2mm\par
   $\diamondsuit$ $(\xi_0+u\xi_1)(\eta_0+u\eta_1)=\xi_0\eta_0+u(\xi_0\eta_1+\xi_1\eta_0)$.

\vskip 2mm\noindent
  Hence $\mathcal{K}_j$ is a subring of $\mathcal{K}_j+u\mathcal{K}_j$. Furthermore,
$\mathcal{K}_j+u\mathcal{K}_j$ is a free $\mathcal{K}_j$-module with a $\mathcal{K}_j$-basis $\{1,u\}$. Now, we define
$$\varsigma: \mathcal{K}_j^2\rightarrow \mathcal{K}_j+u\mathcal{K}_j
\ {\rm via} \ \varsigma: (a_0,a_1)\mapsto a_0+ua_1 \ (\forall a_0,a_1\in \mathcal{K}_j).$$
One can easily verify that $\varsigma$ is a $\mathcal{K}_j$-module isomorphism from $\mathcal{K}_j^2$
onto $\mathcal{K}_j+u\mathcal{K}_j$. Using this $\mathcal{K}_j$-module isomorphism $\varsigma$, we can determine ideals of
the ring $\mathcal{K}_j+u\mathcal{K}_j$ from $\mathcal{K}_j$-submodules of $\mathcal{K}_j^2$ satisfying certain conditions.

\begin{lemma}\label{le3.2}
Using the notations above, $C$ is an ideal
of the ring $\mathcal{K}_j+u\mathcal{K}_j$ if and only if
there is a unique $\mathcal{K}_j$-submodule $S$ of $\mathcal{K}_j^2$ satisfying the following condition:
\begin{equation}\label{eq5}(0,a_0)\in S, \ \forall (a_0,a_1)\in S
\end{equation}
such that $C=\varsigma(S)$.
\end{lemma}

\begin{proof} Let $C$ be an ideal
of $\mathcal{K}_j+u\mathcal{K}_j$. Since $\mathcal{K}_j$ is a subring
of $\mathcal{K}_j+u\mathcal{K}_j$, we see that $C$ is a
$\mathcal{K}_j$-submodule of $\mathcal{K}_j+u\mathcal{K}_j$ satisfying
$u\xi\in C$ for any $\xi\in C$. Now, let
$S=\{(a_0,a_1)\mid a_0+ua_1\in C\}=\varsigma^{-1}(C)$. It is obvious that
$S$ is a $\mathcal{K}_j$-submodule of $\mathcal{K}_j^2$ satisfying $C=\varsigma(S)$. Moreover,
for any $(a_0,a_1)\in S$, i.e. $a_0+ua_1\in C$, by $u^2=0$ it follows
that
$ua_0=u(a_0+ua_1)\in C$. Hence
 $(0,a_0)\in S$.

\par
  Conversely, let $C=\varsigma(S)$ and $S$ be a $\mathcal{K}_j$-submodule of $\mathcal{K}_j^2$ satisfying
Condition (\ref{eq5}). For any $b_0,b_1\in \mathcal{K}_j$ and $a_0+ua_1\in C$ where $(a_0,a_1)\in S$,
by $(0,a_0)\in S$ we have $b_0(a_0,a_1)+b_1(0,a_0)\in S$. On the other hand, we have
$(b_0+ub_1)(a_0+ua_1)=b_0a_0+u(b_0a_1+b_1a_0)$ in $\mathcal{K}_j+u\mathcal{K}_j$. This implies
$(b_0+ub_1)(a_0+ua_1)=\varsigma(b_0a_0,b_0a_1+b_1a_0)=\varsigma(b_0(a_0,a_1)+b_1(0,a_0))\in \varsigma(S)=C$.
Hence $C$ is an ideal of $\mathcal{K}_j+u\mathcal{K}_j$.
\end{proof}

    We notice that $\mathcal{K}_j$-submodules of $\mathcal{K}_j^2$ are called
\textit{linear codes} over the finite chain ring $\mathcal{K}_j$ of length $2$.  Let $S$ be a linear code over $\mathcal{K}_j$ of length $2$. By \cite{s18} Definition 3.1, a matrix $G$ is called a \textit{generator matrix} for $S$ if every codeword in $S$
is a $\mathcal{K}_j$-linear combination of the row vectors of $G$ and
any row vector of $G$ can not be written as a $\mathcal{K}_j$-linear combination of the other row vectors of $G$.
For linear codes over $\mathcal{K}_j$ of length $2$ and their generator matrices, we list the following lemmas.

   \begin{lemma} (cf. \cite{s6} Lemma 2.2) The number of linear codes
over the finite chain ring $\mathcal{K}_j$ of length $2$ is equal to
$\sum_{j=0}^{p^s}(2j+1)|\mathcal{K}_j/\langle f_j(x)\rangle|^{p^s-j}=\sum_{j=0}^{p^s}(2j+1)p^{m(p^s-j)d_j}$.
\end{lemma}

\begin{lemma}\label{le3.4}
 (cf. \cite{s6} Example 2.5) Every linear code $S$ over
$\mathcal{K}_j$ of length $2$ has one and only one of the following matrices $G$ as their generator matrices:

\vskip 2mm \par
(i) $G=(1,a(x))$, $a(x)\in \mathcal{K}_j$.

\vskip 2mm \par
(ii) $G=(f_j(x)^k,f_j(x)^ka(x))$, $a(x)\in \mathcal{K}_j/\langle f_j(x)^{p^s-k}\rangle$ and $1\leq k\leq p^s-1$.

\vskip 2mm \par
(iii) $G=(f_j(x) b(x),1)$, $b(x)\in \mathcal{K}_j/\langle f_j(x)^{p^s-1}\rangle$.

\vskip 2mm \par
(iv) $G=(f_j(x)^{k+1}b(x),f_j(x)^k)$, $b(x)\in \mathcal{K}_j/\langle f_j(x)^{p^s-k-1}\rangle$ and $1\leq k\leq p^s-1$.

\vskip 2mm \par
(v) $G=\left(\begin{array}{cc}f_j(x)^k & 0\cr
0 &f_j(x)^{k}\end{array}\right)$, $0\leq k\leq p^s$.

\vskip 2mm \par
(vi) $G=\left(\begin{array}{cc}1 & c(x)\cr
0 &f_j(x)^{t}\end{array}\right)$,  $c(x)\in \mathcal{K}_j/\langle f_j(x)^t\rangle$ and $1\leq t\leq p^s-1$.

\vskip 2mm \par
(vii) $G=\left(\begin{array}{cc}f_j(x)^k & f_j(x)^kc(x)\cr
0 &f_j(x)^{k+t}\end{array}\right)$, $c(x)\in \mathcal{K}_j/\langle f_j(x)^t\rangle$, $1\leq t\leq p^s-k-1$
and $1\leq k\leq p^s-2$.

\vskip 2mm \par
(viii) $G=\left(\begin{array}{cc}c(x) & 1\cr
f_j(x)^{t} & 0\end{array}\right)$,  $c(x)\in f_j(x)(\mathcal{K}_j/\langle f_j(x)^t\rangle)$ and $1\leq t\leq p^s-1$.

\vskip 2mm \par
(ix) $G=\left(\begin{array}{cc}f_j(x)^kc(x) & f_j(x)^k\cr
f_j(x)^{k+t} & 0\end{array}\right)$, $c(x)\in f_j(x)(\mathcal{K}_j/\langle f_j(x)^t\rangle)$, $1\leq t\leq p^s-k-1$
and $1\leq k\leq p^s-2$.
\end{lemma}

   For any vector $(\xi_0,\xi_1)\in \mathcal{K}_j^2$, we define the \textit{$f_j(x)$-degree} of $(\xi_0,\xi_1)$
by
$$\|(\xi_0,\xi_1)\|_{f_j(x)}={\rm min}\{\|\xi_0\|_{f_j(x)},\|\xi_1\|_{f_j(x)}\}.$$

\begin{lemma}\label{le3.5}(cf. \cite{s18} Proposition 3.2 and Theorem 3.5) Let $S$ be a nonzero linear code over $\mathcal{K}_j$ of length $2$, and $G$ be a generator matrix of $S$ with row vectors $G_1,\ldots,G_\rho$ satisfying
$\|G_j\|_{f_j(x)}=t_j$, where $0\leq t_1\leq\ldots\leq t_\rho\leq p^s-1$.

\vskip 2mm \par
   (i) Every codeword in $S$ can be uniquely expressed as $\sum_{j=1}^\rho b_j(x)G_j$ with $b_j(x)\in \mathcal{K}_j/\langle f_j(x)^{p^s-t_j}\rangle$ for all $j=1,\ldots,\rho$.

\vskip 2mm \par
   (ii) The number of codewords in $S$ is equal to
\begin{center}
   $|S|=|\mathcal{K}_j/\langle f_j(x)\rangle|^{\sum_{j=1}^\rho(p^s-t_j)}=|\mathcal{T}_j|^{\sum_{j=1}^\rho(p^s-t_j)}
=p^{md_j\sum_{j=1}^\rho(p^s-t_j)}$.
\end{center}
\end{lemma}

   The following lemma will be needed in the proof of Lemma \ref{le3.7}.

\begin{lemma} \label{le3.6} (cf. \cite{s19} Lemma 2.2 and Corollary 2.3]) Using the notations above, we have the following conclusions.

\vskip 2mm \par
   (i) For any $\xi\in \mathcal{K}_j$ with $\xi\neq 0$, there is a unique integer $k$, $0\leq k\leq p^s-1$, such
that $\xi=f_j(x)^kc(x)$ for some $c(x)\in \mathcal{K}_j^{\times}$. In this case, $\|\mathcal{K}_j\|_{f_j(x)}=k$ and $c(x)$ is unique
modulo $f_j(x)^{p^s-k}$, i.e. $c(x)\in (\mathcal{K}_j/\langle f_j(x)^{p^s-k}\rangle)^{\times}$.

\vskip 2mm \par
   (ii) Let $1\leq t\leq l\leq p^s$ and $\xi\in \mathcal{K}_j$.
Then $f_j(x)^t\xi\in f_j(x)^l \mathcal{K}_j$ if and only if $\|\xi\|_{f_j(x)}\geq l-t$, i.e. $\xi\in f_j(x)^{l-t}\mathcal{K}_j$.
\end{lemma}

   Then we determine linear codes over $\mathcal{K}_j$ of length $2$, i.e. $\mathcal{K}_j$-submodules of
$\mathcal{K}_j^2$, satisfying Condition (\ref{eq5}) in Lemma \ref{le3.2}.
For any rational number $h$, let $\lceil h\rceil$ be the lest integer greater than or equal to $h$, i.e.
$\lceil h\rceil={\rm min}\{l\in \mathbb{Z} \mid l\geq h\}$. For example, we have $\lceil -\frac{1}{2}\rceil=0$, $\lceil \frac{1}{2}\rceil=1$ and $\lceil \frac{5}{2}\rceil=3$.

\begin{lemma}\label{le3.7}
Using the notations above, every linear code $S$ over
$\mathcal{K}_j$ of length $2$ satisfying Condition (\ref{eq5}) in Lemma \ref{le3.2}
has one and only one of the following matrices $G$ as their generator matrices:

\vskip 2mm \par
   (I) $G=(f_j(x) b(x),1)$,
$b(x)\in f_j(x)^{\lceil \frac{p^s}{2}\rceil-1}(\mathcal{K}_j/\langle f_j(x)^{p^s-1}\rangle)$.

\vskip 2mm \par
   (II) $G=(f_j(x)^{k+1}b(x),f_j(x)^k)$, $b(x)\in f_j(x)^{\lceil\frac{p^s-k}{2}\rceil-1}(\mathcal{K}_j/\langle f_j(x)^{p^s-k-1}\rangle)$ and $1\leq k\leq p^s-1$.

\vskip 2mm \par
   (III) $G=\left(\begin{array}{cc}f_j(x)^k & 0\cr
0 &f_j(x)^{k}\end{array}\right)$, $0\leq k\leq p^s$.

\vskip 2mm \par
   (IV) $G=\left(\begin{array}{cc}c(x) & 1\cr
f_j(x)^{t} & 0\end{array}\right)$, $c(x)\in f_j(x)^{\lceil\frac{t}{2}\rceil}(\mathcal{K}_j/\langle f_j(x)^t\rangle)$
and $1\leq t\leq p^s-1$.

\vskip 2mm \par
   (V) $G=\left(\begin{array}{cc}f_j(x)^kc(x) & f_j(x)^k\cr
f_j(x)^{k+t} & 0\end{array}\right)$, $c(x)\in f_j(x)^{\lceil\frac{t}{2}\rceil}(\mathcal{K}_j/\langle f_j(x)^t\rangle)$, $1\leq t\leq p^s-k-1$
and $1\leq k\leq p^s-2$.
\end{lemma}

\begin{proof} For notations simplicity, we denote $\pi=f_j(x)$
and use lowercase letters, say $a,b,c,\ldots$, to denote elements of
the finite chain ring $\mathcal{K}_j$.
   By Lemma \ref{le3.4} we only need to consider the following nine cases.

\vskip 2mm\par
   (i) $G=(1,a)$, where $a\in \mathcal{K}_j$. Suppose that $S$ satisfies Condition (\ref{eq5}). Then
$(0,1)\in S$. Since $G$ is the generator matrix of $S$, there exists $b\in \mathcal{K}_j$ such that $(0,1)=b(1,a)=(b,ba)$,
i.e. $0=b$ and $1=ba$, which is impossible. Hence $S$ does not satisfy Condition (\ref{eq5}) in this case.

\vskip 2mm\par
   (ii) $G=(\pi^k,\pi^ka)$, where $a\in \mathcal{K}_j/\langle \pi^{p^s-k}\rangle$ and $1\leq k\leq p^s-1$.
Suppose that $S$ satisfies Condition (\ref{eq5}). Then
$(0,\pi^k)\in S$. So there exists $b\in \mathcal{K}_j$ such that $(0,\pi^k)=b(\pi^k,\pi^ka)=(\pi^kb,\pi^kab)$,
which implies $0=\pi^kb$ and $\pi^k=\pi^kba$. Hence $\pi^k=0$. But $1\leq k\leq p^s-1$, we get a contradiction. Hence $S$ does not satisfy Condition (\ref{eq5}) in this case.

\vskip 2mm\par
   (iii) $G=(\pi b,1)$, where $b\in \mathcal{K}_j/\langle \pi^{p^s-1}\rangle$. Then $S$ satisfies Condition (\ref{eq5}) if and only if
there exists $a\in \mathcal{K}_j$ such that $(0,\pi b)=a(\pi b,1)=(\pi ba,a)$, i.e. $0=\pi ba$ and $\pi b=a$, which are equivalent to
that $b$ satisfies $\pi^2b^2=0$. By Lemma \ref{le3.6}(ii) we see that $b\in \mathcal{K}_j/\langle \pi^{p^s-1}\rangle$ satisfying  $\pi^2b^2=0$ if and only if $b^2\in \pi^{p^s-2}A$, i.e., $2\|b\|_\pi=\|b^2\|_\pi\geq p^s-2$, and hence $\|b\|_\pi\geq \lceil \frac{1}{2}(p^s-2)\rceil
=\lceil \frac{p^s}{2}\rceil-1$.
Therefore, $b\in \pi^{\lceil \frac{p^s}{2}\rceil-1} \mathcal{K}_j\cap (\mathcal{K}_j/\langle \pi^{p^s-1}\rangle)=\pi^{\lceil \frac{p^s}{2}\rceil-1}(\mathcal{K}_j/\langle \pi^{p^s-1}\rangle)$.

\vskip 2mm\par
  (iv) $G=(\pi^{k+1} b,\pi^k)$, where $b\in \mathcal{K}_j/\langle \pi^{p^s-k-1}\rangle$ and $1\leq k\leq p^s-1$. Then $S$ satisfies Condition (\ref{eq5}) if and only if
there exists $a\in \mathcal{K}_j$ such that $(0,\pi^{k+1} b)=a(\pi^{k+1} b,\pi^k)=(\pi^{k+1} ab,\pi^ka)$, i.e., $0=\pi^{k+1} ab$ and $\pi^{k+1} b=\pi^ka$, which are equivalent to
that $b$ satisfies $\pi^{k+2}b^2=0$. Then by
$3\leq k+2\leq p^s+1$, we have one of the following two subcases:

\par
   (iv-1) When $k+2\geq p^s$, i.e., $k=p^s-2$ or $p^s-1$, then $\pi^{k+2}=0$ and hence $\pi^{k+2}b^2=0$ for every $b\in \mathcal{K}_j/\langle \pi^{p^s-k-1}\rangle$.

\par
   (iv-2) When $k+2\leq p^s-1$, i.e., $k\leq p^s-3$ (and $p^s\geq 4$), then $b\in \mathcal{K}_j/\langle \pi^{p^s-k-1}\rangle$ satisfying  $\pi^{k+2}b^2=0$ if and only if $b^2\in \pi^{p^s-k-2} \mathcal{K}_j$.
From this and by Lemma \ref{le3.6}(ii), we deduce that $\|b\|_\pi\geq\lceil\frac{1}{2}(p^s-k-2)\rceil$. So
 $b\in \pi^{\lceil\frac{1}{2}(p^s-k-2)\rceil}(\mathcal{K}_j/\langle \pi^{p^s-k-1}\rangle)$.

\vskip 2mm \par
   (v) $G=\left(\begin{array}{cc}\pi^k & 0\cr
0 &\pi^{k}\end{array}\right)$, where $0\leq k\leq p^s$. It is clear that $S$ satisfies Condition (\ref{eq5})
for any $k$.

\vskip 2mm \par
   (vi) $G=\left(\begin{array}{cc}1 & c\cr
0 &\pi^{t}\end{array}\right)$, where $c\in \mathcal{K}_j/\langle \pi^t\rangle$ and $1\leq t\leq p^s-1$. Suppose
that $S$ satisfies Condition (\ref{eq5}). Then there exist $a,b\in \mathcal{K}_j$ such that $(0,1)=a(1,c)+b(0,\pi^t)=(a,ac+\pi^tb)$,
i.e. $0=a$ and $1=ac+\pi^tb$, which imply $1=\pi^tb$, and we get a contradiction.
Hence $S$ does not satisfy Condition (\ref{eq5}) in this case.

\vskip 2mm \par
   (vii) $G=\left(\begin{array}{cc}\pi^k & \pi^kc\cr
0 &\pi^{k+t}\end{array}\right)$, where $c\in \mathcal{K}_j/\langle \pi^t\rangle$, $1\leq t\leq p^s-k-1$
and $1\leq k\leq p^s-2$. Suppose
that $S$ satisfies Condition (\ref{eq5}). Then there exist $a,b\in \mathcal{K}_j$ such that $(0,\pi^k)=a(\pi^k,\pi^kc)+b(0,\pi^{k+t})=(\pi^ka,\pi^kac+\pi^{k+t}b)$,
i.e. $0=\pi^ka$ and $\pi^k=\pi^kac+\pi^{k+t}b$, which imply $\pi^k=\pi^{k+t}b$, and we get a contradiction.
Hence $S$ does not satisfy Condition (\ref{eq5}) in this case.

\vskip 2mm \par
   (viii) $G=\left(\begin{array}{cc}c & 1\cr
\pi^{t} & 0\end{array}\right)$, where $c\in \pi(\mathcal{K}_j/\langle \pi^t\rangle)$ and $1\leq t\leq p^s-1$. It is clear that $(0,\pi^t)=\pi^t(c,1)-c(\pi^t,0)\in S$. Then $S$ satisfies Condition (\ref{eq5}) if and only if
there exist $a,b\in \mathcal{K}_j$ such that $(0,c)=a(c,1)+b(\pi^t,0)=(ac+\pi^tb,a)$, i.e. $0=ac+\pi^tb$ and $c=a$,
which are equivalent to that $c^2=-\pi^tb\in \pi^t\mathcal{K}_j$, i.e. $2\|c\|_\pi\geq t$. Therefore,
$c\in \pi^{\lceil\frac{t}{2}\rceil}(\mathcal{K}_j/\langle \pi^t\rangle)$ by Lemma \ref{le3.6}(ii).

\vskip 2mm \par
   (ix) $G=\left(\begin{array}{cc}\pi^kc & \pi^k\cr
\pi^{k+t} & 0\end{array}\right)$, where $c\in \pi(\mathcal{K}_j/\langle \pi^t\rangle)$, $1\leq t\leq p^s-k-1$
and $1\leq k\leq p^s-2$. Obviously, $(0,\pi^{k+t})=\pi^t(\pi^kc,\pi^k)-c(\pi^{k+t},0)\in S$. Then $S$ satisfies Condition (\ref{eq5}) if and only if
there exist $a,b\in \mathcal{K}_j$ such that $(0,\pi^kc)=a(\pi^kc,\pi^k)+b(\pi^{k+t},0)=(\pi^kac+\pi^{k+t}b,\pi^ka)$, i.e. $0=\pi^kac+\pi^{k+t}b$ and $\pi^kc=\pi^ka$,
which are equivalent to that $\pi^kc^2=-\pi^{k+t}b\in \pi^{k+t}\mathcal{K}_j$.
Then by Lemma \ref{le3.6}(ii) we deduce that $c^2\in \pi^t\mathcal{K}_j$. Hence $c\in \pi^{\lceil\frac{t}{2}\rceil}(\mathcal{K}_j/\langle \pi^t\rangle)$.
\end{proof}

    Now, we can list all distinct ideals of the
ring $\mathcal{K}_j+u\mathcal{K}_j$.

\begin{theorem}\label{th3.8}
  Using the notations above, all distinct ideals $C_j$ of the
ring $\mathcal{K}_j+u\mathcal{K}_j$ $(u^2=0)$ are given by one of the following five cases:

\vskip 2mm \par
 (I) $p^{\left(p^s-\lceil\frac{p^s}{2}\rceil\right)md_j}$ ideals:
$$C_j=\left\langle f_j(x)b(x)+u\right\rangle \ {\rm with} \ |C_j|=p^{md_jp^s},$$
where $b(x)\in f_j(x)^{\lceil \frac{p^s}{2}\rceil-1}(\mathcal{K}_j/\langle f_j(x)^{p^s-1}\rangle)$.

\vskip 2mm \par
 (II) $\sum_{k=1}^{p^s-1}p^{\left(p^s-k-\lceil\frac{1}{2}(p^s-k)\rceil\right)md_j}$ ideals:
$$C_j=\left\langle f_j(x)^{k+1}b(x)+uf_j(x)^k\right\rangle \ {\rm with} \ |C_j|=p^{md_j(p^s-k)},$$
where $b(x)\in f_j(x)^{\lceil \frac{p^s-k}{2}\rceil-1}(\mathcal{K}_j/\langle f_j(x)^{p^s-k-1}\rangle)$ and $1\leq k\leq p^s-1$.

\vskip 2mm \par
 (III) $p^s+1$ ideals:
$C_j=\left\langle f_j(x)^k\right\rangle \ {\rm with} \ |C_j|=p^{2md_j(p^s-k)}, \ 0\leq k\leq p^s.$

\vskip 2mm \par
 (IV) $\sum_{t=1}^{p^s-1}p^{\left(t-\lceil\frac{t}{2}\rceil\right)md_j}$ ideals:
$$C_j=\left\langle f_j(x)b(x)+u,f_j(x)^{t}\right\rangle \ {\rm with} \ |C_j|=p^{md_j(2p^s-t)},$$
where
$b(x)\in f_j(x)^{\lceil \frac{t}{2}\rceil-1}(\mathcal{K}_j/\langle f_j(x)^{t-1}\rangle)$,
$1\leq t\leq p^s-1$.

\vskip 2mm \par
 (V) $\sum_{k=1}^{p^s-2}\sum_{t=1}^{p^s-k-1}p^{\left(t-\lceil\frac{t}{2}\rceil\right)md_j}$ ideals:
$$C_j=\left\langle f_j(x)^{k+1}b(x)+uf_j(x)^k,f_j(x)^{k+t}\right\rangle \ {\rm with} \ |C_j|=p^{md_j(2p^s-2k-t)},$$
where
$b(x)\in f_j(x)^{\lceil \frac{t}{2}\rceil-1}(\mathcal{K}_j/\langle f_j(x)^{t-1}\rangle)$,
$1\leq t\leq p^s-k-1$ and $1\leq k\leq p^s-2$.

\vskip 2mm \par
 Moreover, let $N_{(p^m,d_j,p^s)}$ be the number of ideals in $\mathcal{K}_j+u\mathcal{K}_j$. Then
$$N_{(p^m,d_j,p^{s})}=\left\{\begin{array}{ll}\sum_{i=0}^{2^{s-1}}(1+4i)2^{(2^{s-1}-i)md_j}, & {\rm if} \ p=2; \cr & \cr
                                        \sum_{i=0}^{\frac{p^s-1}{2}}(3+4i)p^{(\frac{p^s-1}{2}-i)md_j}, & {\rm if} \ p \ {\rm is} \ {\rm odd}. \end{array}\right.$$
\end{theorem}

\begin{proof}
Let $C_j$ be a nontrivial ideal of $\mathcal{K}_j+u\mathcal{K}_j$. By Lemma \ref{le3.2}, there is a unique
$\mathcal{K}_j$-submodule $S_j$ of $\mathcal{K}_j^2$ satisfying Condition (\ref{eq5}):
$(0,a_0)\in S_j$ for all $(a_0,a_1)\in S_j$, such that $C_j=\varsigma(S_j)$. By
Lemma \ref{le3.7}, we see that $S_j$ has one and
only one of the following matrices $G_j$ as their generator matrices:
\vskip 2mm \par
   (i) $G_j=(f_j(x) b(x),1)$,
$b(x)\in f_j(x)^{\lceil \frac{p^s}{2}\rceil-1}(\mathcal{K}_j/\langle f_j(x)^{p^s-1}\rangle)$.

\vskip 2mm \par
   (ii) $G_j=(f_j(x)^{k+1}b(x),f_j(x)^k)$, $b(x)\in f_j(x)^{\lceil\frac{p^s-k}{2}\rceil-1}(\mathcal{K}_j/\langle f_j(x)^{p^s-k-1}\rangle)$ and $1\leq k\leq p^s-1$.

\vskip 2mm \par
   (iii) $G_j=\left(\begin{array}{cc}f_j(x)^k & 0\cr
0 &f_j(x)^{k}\end{array}\right)$, $0\leq k\leq p^s$.

\vskip 2mm \par
   (iv) $G_j=\left(\begin{array}{cc}c(x) & 1\cr
f_j(x)^{t} & 0\end{array}\right)$, $c(x)\in f_j(x)^{\lceil\frac{t}{2}\rceil}(\mathcal{K}_j/\langle f_j(x)^t\rangle)$
and $1\leq t\leq p^s-1$.

\vskip 2mm \par
   (v) $G_j=\left(\begin{array}{cc}f_j(x)^kc(x) & f_j(x)^k\cr
f_j(x)^{k+t} & 0\end{array}\right)$, $c(x)\in f_j(x)^{\lceil\frac{t}{2}\rceil}(\mathcal{K}_j/\langle f_j(x)^t\rangle)$, $1\leq t\leq p^s-k-1$
and $1\leq k\leq p^s-2$.

\vskip 2mm \par
   (I) Let $G_j$ be given in (i). By Lemma \ref{le3.2} we have
$C_j=\varsigma(S_j)=\langle \varsigma(f_j(x) b(x),1)\rangle=\langle f_j(x) b(x)+u\rangle$. As $\|(f_j(x) b(x),1)\|_{f_j(x)}=0$,
by Lemma \ref{le3.5}(ii) the number of elements in $S_j$ is equal to $|S_j|=p^{md_j(p^s-0)}=p^{md_jp^s}$. Hence
$|C_j|=|S_j|=p^{md_jp^s}$ by Lemma \ref{le3.2}.

\par
  In this case, by Lemma \ref{le3.1}(v) and Lemma \ref{le3.1} we deduce that the number of ideals is equal to
$|f_j(x)^{\lceil \frac{p^s}{2}\rceil-1}(\mathcal{K}_j/\langle f_j(x)^{p^s-1}\rangle|
=p^{md_j(p^s-\lceil \frac{p^s}{2}\rceil)}$.

\par
   Case (II) can be proved similarly as that of (I).

\par
   (III) Let $G_j$ be given in (iii). By Lemma \ref{le3.2}, it follows that
\begin{center}
$C_j=\varsigma(S_j)=\langle \varsigma(f_j(x)^k,0),\varsigma(0,f_j(x)^k)\rangle=\langle f_j(x)^k,uf_j(x)^k\rangle=\langle f_j(x)^k\rangle$.
\end{center}
 As $\|(f_j(x)^k,0)\|_{f_j(x)}=\|(0,f_j(x)^k)\|_{f_j(x)}=k$,
by Lemmas \ref{le3.2} and \ref{le3.5} we deduce that $|C_j|=|S_j|=p^{md_j((p^s-k)+(p^s-k))}=p^{2md_j(p^s-k)}$.

\par
   (IV) Let $G_j$ be given in (iv). By Lemma \ref{le3.2} we have \\
$C_j=\varsigma(S_j)=\langle \varsigma(c(x),1),\varsigma(f_j(x)^t,0)\rangle=\langle c(x)+u,f_j(x)^t\rangle$.
 As $\|(c(x),1)\|_{f_j(x)}=0$ and $\|(f_j(x)^t,0)\|_{f_j(x)}=t$,
by Lemmas \ref{le3.2} and \ref{le3.5} we deduce that $|C_j|=|S_j|=p^{md_j((p^s-0)+(p^s-t))}=p^{md_j(2p^s-t)}$.

\par
   In this case, by Lemma \ref{le3.1}(v) and Lemma \ref{le3.2} we deduce that the number of ideals is equal to
$|f_j(x)^{\lceil \frac{t}{2}\rceil}(\mathcal{K}_j/\langle f_j(x)^{t}\rangle)|=p^{md_j(t-\lceil \frac{t}{2}\rceil)}$.

\par
  Furthermore, by $\lceil \frac{t}{2}\rceil\geq 1$ we can write
$c(x)=f_j(x)b(x)$ for any \\
$c(x)\in f_j(x)^{\lceil \frac{t}{2}\rceil}(\mathcal{K}_j/\langle f_j(x)^{t}\rangle)$,  where $b(x)\in f_j(x)^{\lceil \frac{t}{2}\rceil-1}(\mathcal{K}_j/\langle f_j(x)^{t-1}\rangle)$ and $b(x)$ is uniquely
determined by $c(x)$.

\par
   Case (V) can be proved similarly as that of (IV).

\par
  As stated above, we conclude that the number of ideals of $\mathcal{K}_j+u\mathcal{K}_j$ is equal to
$N_{(p^m,d_j,p^s)}=1+p^s+\sum_{k=0}^{p^s-1}p^{\left(p^s-k-\lceil\frac{1}{2}(p^s-k)\rceil\right)md_j}+
\sum_{k=0}^{p^s-2}\sum_{t=1}^{p^s-k-1}p^{\left(t-\lceil\frac{t}{2}\rceil\right)md_j}.$

\end{proof}

   Therefore, by Theorems \ref{th2.5} and \ref{th3.8} we can determine all distinct $\lambda$-constacyclic codes over $R$ of length $np^s$.

\begin{corollary}\label{co3.9}
Using the notations above, every $\lambda$-constacyclic
code ${\mathcal C}$ over $R$ of length $np^s$ can be constructed by the following two steps:

\vskip2 mm\par
  (i) For each $j=1,\ldots,r$, choose an ideal $C_j$ of
$\mathcal{K}_j+u\mathcal{K}_j$ listed in Theorem \ref{th3.8}.

\vskip2 mm\par
  (ii) Set ${\mathcal C}=\bigoplus_{j=1}^r\varepsilon_j(x)C_j$ (mod $x^{p^sn}-\lambda$).

\vskip 2mm \noindent
The number of codewords in ${\mathcal C}$ is
equal to $|{\mathcal C}|=\prod_{j=1}^r|C_j|$.

\par
   Therefore, the number of $\lambda$-constacyclic
codes over $R$ of length $np^s$ is equal to $\prod_{j=1}^rN_{(p^m,d_j,p^s)}$.
\end{corollary}

   Using the notations of Corollary \ref{co3.9}(ii), ${\mathcal C}=\bigoplus_{j=1}^r\varepsilon_j(x)C_j$ is called the \textit{canonical form decomposition} of the $\lambda$-constacyclic code ${\mathcal C}$.
   As the end of this section, we consider the special situation of $r=1$, i.e. $x^n-\lambda_0$ is irreducible
in $\mathbb{F}_{p^m}[x]$.

\begin{corollary}\label{co3.10}
\textit{Let $x^n-\lambda_0$ be an irreducible polynomial in
$\mathbb{F}_{p^m}[x]$ and denote $\mathcal{A}=\mathbb{F}_{p^m}[x]/\langle x^{np^s}-\lambda_0^{p^s}\rangle$.
Then all distinct $\lambda_0^{p^s}$-constacyclic codes over $\mathbb{F}_{p^m}+u\mathbb{F}_{p^m}$ of length $np^s$, i.e.
all ideals of the ring $\mathcal{A}+u\mathcal{A}$, are given by
the following five cases}:

\vskip 2mm \par
 (I) $p^{\left(p^s-\lceil\frac{p^s}{2}\rceil\right)mn}$ codes:
$$\mathcal{C}=\langle (x^n-\lambda_0)b(x)+u\rangle \ {\rm with} \ |\mathcal{C}|=p^{mnp^s},$$
where $b(x)\in (x^n-\lambda_0)^{\lceil \frac{p^s}{2}\rceil-1}(\mathcal{A}/\langle (x^n-\lambda_0)^{p^s-1}\rangle)$.

\vskip 2mm \par
 (II) $\sum_{k=1}^{p^s-1}p^{\left(p^s-k-\lceil\frac{1}{2}(p^s-k)\rceil\right)mn}$ codes:
$$\mathcal{C}=\langle (x^n-\lambda_0)^{k+1}b(x)+u(x^n-\lambda_0)^k\rangle \ {\rm with} \ |\mathcal{C}|=p^{mn(p^s-k)},$$
where $b(x)\in (x^n-\lambda_0)^{\lceil \frac{p^s-k}{2}\rceil-1}(\mathcal{A}/\langle (x^n-\lambda_0)^{p^s-k-1}\rangle)$ and $1\leq k\leq p^s-1$.

\vskip 2mm \par
 (III) $p^s+1$ codes:
$\mathcal{C}=\langle (x^n-\lambda_0)^k\rangle \ {\rm with} \ |\mathcal{C}|=p^{2mn(p^s-k)}, \ 0\leq k\leq p^s.$

\vskip 2mm \par
 (IV) $\sum_{t=1}^{p^s-1}p^{\left(t-\lceil\frac{t}{2}\rceil\right)mn}$ codes:
$$\mathcal{C}=\langle (x^n-\lambda_0)b(x)+u,(x^n-\lambda_0)^{t}\rangle \ {\rm with} \ |\mathcal{C}|=p^{mn(2p^s-t)},$$
where
$b(x)\in (x^n-\lambda_0)^{\lceil\frac{t}{2}\rceil-1}(\mathcal{A}/\langle (x^n-\lambda_0)^{t-1}\rangle)$,
$1\leq t\leq p^s-1$.

\vskip 2mm \par
 (V) $\sum_{k=1}^{p^s-2}\sum_{t=1}^{p^s-k-1}p^{\left(t-\lceil\frac{t}{2}\rceil\right)mn}$ codes:
$$\mathcal{C}=\langle (x^n-\lambda_0)^{k+1}b(x)+u(x^n-\lambda_0)^k,(x^n-\lambda_0)^{k+t}\rangle \ {\rm with} \ |\mathcal{C}|=p^{mn(2p^s-2k-t)},$$
where
$b(x)\in (x^n-\lambda_0)^{\lceil\frac{t}{2}\rceil-1}(\mathcal{A}/\langle (x^n-\lambda_0)^{t-1}\rangle)$,
$1\leq t\leq p^s-k-1$ and $1\leq k\leq p^s-2$.

\vskip 2mm \par
  Moreover, the number of $\lambda_0^{p^s}$-constacyclic codes over $\mathbb{F}_{p^m}+u\mathbb{F}_{p^m}$ of length $np^s$ is equal to
$N_{(p^m,n,p^{s})}=\left\{\begin{array}{ll}\sum_{i=0}^{2^{s-1}}(1+4i)2^{(2^{s-1}-i)mn}, & {\rm if} \ p=2; \cr & \cr
                                        \sum_{i=0}^{\frac{p^s-1}{2}}(3+4i)p^{(\frac{p^s-1}{2}-i)mn}, & {\rm if} \ p \ {\rm is} \ {\rm odd}. \end{array}\right.$
\end{corollary}

\begin{proof}
    The conclusions follows by substituting $\mathcal{K}_j$, $f_j(x)$ and $d_j$ by $\mathcal{A}$, $x^n-\lambda_0$ and $n$ in Theorem \ref{th3.8}, respectively.
\end{proof}


\section{Dual codes of $\lambda$-constacyclic codes over $R$ of length $np^s$}

\noindent \par
In this section, we give the dual code of every
$\lambda$-constacyclic code over $R$ of length $np^s$.

\par
  For any polynomial $f(x)=\sum_{i=0}^dc_ix^i\in \mathbb{F}_{p^m}[x]$ of degree $d\geq 1$,
the \textit{reciprocal polynomial} of $f(x)$ is defined as $\widetilde{f}(x)=\widetilde{f(x)}=x^df(\frac{1}{x})=\sum_{i=0}^dc_ix^{d-i}$.
 $f(x)$ is said to be \textit{self-reciprocal} if $\widetilde{f}(x)=\delta f(x)$ for some $\delta\in \mathbb{F}_{p^m}^{\times}$.
It is known that $\widetilde{\widetilde{f}(x)}=f(x)$ if $f(0)\neq 0$, and $\widetilde{f(x)g(x)}=\widetilde{f}(x)\widetilde{g}(x)$ for
any monic polynomials $f(x), g(x)\in\mathbb{F}_{p^m}[x]$ with positive degrees satisfying $f(0),g(0)\in \mathbb{F}_{p^m}^{\times}$.

\par
   By $\lambda_0^{p^s}=\lambda$, we have $(\lambda_0^{-1})^{p^s}=\lambda^{-1}$. Then using the notions and conclusions in Section 2, we have
$$x^n-\lambda_0^{-1}=-\lambda_0^{-1}\widetilde{(x^n-\lambda_0)}=-\lambda_0^{-1}\widetilde{f}_1(x)\ldots\widetilde{f}_r(x),$$
\begin{equation}\label{eq6}x^{np^s}-\lambda^{-1}=-\lambda^{-1}\widetilde{(x^{np^s}-\lambda)}=-\lambda^{-1}\widetilde{f}_1(x)^{p^s}\ldots\widetilde{f}_r(x)^{p^s},
\end{equation}
In the following, we adopt the following notations and definitions.

\begin{notation}\label{no4.1}Let $\lambda, \lambda_0\in \mathbb{F}_{p^m}^{\times}$ satisfying $\lambda_0^{p^s}=\lambda$. For any $1\leq j\leq r$ we denote

\vskip 2mm \par
  $\bullet$ $\widehat{\mathcal{A}}=\mathbb{F}_{p^m}[x]/\langle x^{np^s}-\lambda^{-1}\rangle$,
  $\widehat{\mathcal{A}}[u]/\langle u^2\rangle=\widehat{\mathcal{A}}+u\widehat{\mathcal{A}}$ ($u^2=0$);

\vskip 2mm \par
  $\bullet$ $\widehat{\mathcal{K}}_j=\mathbb{F}_{p^m}[x]/\langle \widetilde{f}_j(x)^{p^s}\rangle$,
  $\widehat{\mathcal{K}}_j[u]/\langle u^2\rangle=\widehat{\mathcal{K}}_j+u\widehat{\mathcal{K}}_j$ ($u^2=0$);

\vskip 2mm \par
 $\bullet$   $\widehat{\Psi}: \widehat{\mathcal{A}}+u\widehat{\mathcal{A}}\rightarrow R[x]/\langle x^{np^s}-\lambda^{-1}\rangle$
via
\begin{center}
$\widehat{\Psi}: g_0(x)+ug_1(x)\mapsto \sum_{i=0}^{np^s-1}(g_{i,0}+ug_{i,1})x^i$
\end{center}
($\forall g_k(x)=\sum_{i=0}^{np^s-1}g_{i,k}x^i
\in \widehat{\mathcal{A}}$ with $g_{i,k}\in \mathbb{F}_{p^m}$, $0\leq i\leq np^s-1$ and $k=0,1$).
\end{notation}

\par
   Similar to Lemma \ref{le2.1}, it can be easily verified that $\widehat{\Psi}$ is a ring isomorphism from $\widehat{\mathcal{A}}+u\widehat{\mathcal{A}}$
onto $\mathcal{R}_{\lambda^{-1}}=R[x]/\langle x^{np^s}-\lambda^{-1}\rangle$.
Then we will identify $\widehat{\mathcal{A}}+u\widehat{\mathcal{A}}$
with $\mathcal{R}_{\lambda^{-1}}$ under $\widehat{\Psi}$ in the rest of the paper.
As $x^{np^s}=\lambda^{-1}$ in the ring $\widehat{\mathcal{A}}$, we have
$$x^{-1}=\lambda x^{np^s-1} \ {\rm in} \ \widehat{\mathcal{A}}\subset\mathcal{R}_{\lambda^{-1}}.$$
   Now, we define a map $\tau:\mathcal{A}\rightarrow \widehat{\mathcal{A}}$ via
\begin{center}
$\tau: a(x)\mapsto a(x^{-1})=a(\lambda x^{np^s-1})$ (mod $x^{np^s}-\lambda^{-1}$), for all $a(x)\in \mathcal{A}$.
\end{center}
Then one can easily verify that
$\tau$ is a ring isomorphism from $\mathcal{A}$ onto $\widehat{\mathcal{A}}$ and can be extended to a ring isomorphism
from $\mathcal{A}+u\mathcal{A}$ onto $\widehat{\mathcal{A}}+u\widehat{\mathcal{A}}$ in the natural way that
\begin{center}
$\tau: \rho(x)\mapsto \rho(x^{-1})=a(x^{-1})+ub(x^{-1})$ ($\forall \rho(x)=a(x)+ub(x)$, $a(x),b(x)\in \mathcal{A}$).
\end{center}
Now, let $1\leq j\leq r$. By Equation (\ref{eq6}), $\widetilde{f}_j(x)^{p^s}$ is a divisor of $x^{np^s}-\lambda^{-1}$
in $\mathbb{F}_{p^m}[x]$. This implies $x^{np^s}\equiv\lambda^{-1}$ (mod $\widetilde{f}_j(x)^{p^s}$). Hence
$x^{-1}=\lambda x^{np^s-1}$ in the rings $\widehat{\mathcal{K}}_j$ and $\widehat{\mathcal{K}}_j+u \widehat{\mathcal{K}}_j$ as well. Moreover, we define

\vskip 2mm \par
  $\bullet$ $\widehat{\varepsilon}_j(x)=\tau(\varepsilon_j(x))
=\varepsilon_j(x^{-1})=\varepsilon_j(\lambda x^{np^s-1}) \ ({\rm mod} \ x^{np^s}-\lambda^{-1});$

\vskip 2mm \par
  $\bullet$  $\widehat{\Phi}: (\widehat{\mathcal{K}}_1+u\widehat{\mathcal{K}}_1)\times\ldots\times(\widehat{\mathcal{K}}_r+u\widehat{\mathcal{K}}_r)\rightarrow \widehat{\mathcal{A}}+u\widehat{\mathcal{A}}$ via
\begin{center}
$\widehat{\Phi}: (\xi_1+u\eta_1,\ldots,\xi_r+u\eta_r)\mapsto
\sum_{j=1}^r\widehat{\varepsilon}_j(x)(\xi_j+u\eta_j)$  (mod $x^{np^s}-\lambda^{-1}$)
\end{center}
($\forall \xi_j,\eta_j\in \widehat{\mathcal{K}}_j$, $j=1,\ldots,r$);

\vskip 2mm \par
$\bullet$ $\tau_j: \mathcal{K}_j+u\mathcal{K}_j\rightarrow \widehat{\mathcal{K}}_j+u\widehat{\mathcal{K}}_j$ via
\begin{center}
$\tau_j: \xi\mapsto a(x^{-1})+ub(x^{-1})=a(\lambda x^{np^s-1})+ub(\lambda x^{np^s-1})$ (mod $\widetilde{f}_j(x)^{p^s}$),
\end{center}
$\forall \xi=a(x)+ub(x)\in \mathcal{K}_j+u\mathcal{K}_j$ with $a(x),b(x)\in \mathcal{K}_j$.

\vskip 2mm \noindent
  Then similar to Lemma \ref{le2.4}, we see that $\widehat{\Phi}$ is a ring isomorphism from $ (\widehat{\mathcal{K}}_1+u\widehat{\mathcal{K}}_1)\times\ldots\times(\widehat{\mathcal{K}}_r+u\widehat{\mathcal{K}}_r)$ onto $\widehat{\mathcal{A}}+u\widehat{\mathcal{A}}$. Moreover, we have the following.

\vskip 3mm \noindent
\begin{lemma}\label{le4.2} Using the notations above, $\tau_j$ is a ring isomorphism from
$\mathcal{K}_j+u\mathcal{K}_j$ onto $\widehat{\mathcal{K}}_j+u\widehat{\mathcal{K}}_j$ satisfying:
$\tau(\varepsilon_j(x)\xi)=\widehat{\varepsilon}_j(x)\tau_j(\xi)$ for all $\xi\in \mathcal{K}_j+u\mathcal{K}_j$.
\end{lemma}

\begin{proof} Since $x^{-1}=\lambda x^{p^sn-1}$ in $\widehat{\mathcal{K}}_j=\mathbb{F}_{p^m}[x]/\langle \widetilde{f}_j(x)^{p^s}\rangle$,
$\tau_j$ is well defined. By
$$f_j(x^{-1})^{p^s}=\lambda x^{p^sn}f_j(x^{-1})^{p^s}=\lambda x^{p^s(n-d_j)}(x^{d_j}f_j(x^{-1}))^{p^s}
=\lambda x^{p^s(n-d_j)}\widetilde{f}_j(x)^{p^s}$$
in $\widehat{\mathcal{K}}_j$ and $\lambda x^{p^s(n-d_j)}\in \widehat{\mathcal{K}}_j^\times$,
we have $\langle f_j(x^{-1})^{p^s}\rangle=\langle \widetilde{f}_j(x)^{p^s}\rangle$.
From this, one can easily verify that
$\tau_j$ is a ring isomorphism from
$\mathcal{K}_j+u\mathcal{K}_j$ onto $\widehat{\mathcal{K}}_j+u\widehat{\mathcal{K}}_j$. Finally, for any
$\xi=a(x)+ub(x)\in \mathcal{K}_j+u\mathcal{K}_j$, where $a(x),b(x)\in \mathcal{K}_j$, by the definitions of
$\tau$ and $\tau_j$ we have that $\tau(\varepsilon_j(x)\xi)=\varepsilon_j(x^{-1})(a(x^{-1})+ub(x^{-1}))=\widehat{\varepsilon}_j(x)\tau_j(\xi)$.
\end{proof}

\begin{lemma}\label{le4.3}
\textit{Let $a=(a_0,a_1,\ldots,a_{p^sn-1}), b=(b_0,b_1,\ldots,b_{p^sn-1})\in R^{np^s}$
where $a_i,b_i\in R$ for all $i=0,1,\ldots,np^s-1$. We denote
$$a(x)=\sum_{i=0}^{np^s-1}a_ix^i\in \mathcal{A}+u\mathcal{A}, \
b(x)=\sum_{i=0}^{np^s-1}b_ix^i\in \widehat{\mathcal{A}}+u\widehat{\mathcal{A}}.$$
Then $[a,b]=\sum_{i=0}^{np^s-1}a_ib_i=0$, if $\tau(a(x))\cdot b(x)=0$ in $\widehat{\mathcal{A}}+u\widehat{\mathcal{A}}$.}
\end{lemma}

\begin{proof} By $x^{np^s}=\lambda^{-1}$ and $x^{-1}=\lambda x^{np^s-1}$ in $\widehat{\mathcal{A}}+u\widehat{\mathcal{A}}=R[x]/\langle x^{np^s}-\lambda^{-1}\rangle$,
we have
$$x^{-i}=(\lambda x^{np^s-1})^i=\lambda^i (x^{np^s})^{i-1}x^{np^s-i}=\lambda x^{np^s-i}
\ (\forall 1\leq i\leq np^s-1),$$
and hence $\tau(a(x))=a(x^{-1})=a_0+\lambda\sum_{i=1}^{np^s-1}a_ix^{np^s-i}$. Therefore,
$\tau(a(x))\cdot b(x)=[a,b]+\sum_{i=1}^{p^sn-1}c_ix^i$ for some $c_1,\ldots,c_{p^sn-1}\in R$.
Hence $[a,b]=0$ when $\tau(a(x))\cdot b(x)=0$ in $\widehat{\mathcal{A}}+u\widehat{\mathcal{A}}$.
\end{proof}

   Now, we can represent the dual code of each $\lambda$-constacyclic code over $R$ of length $np^s$
from its canonical form decomposition.

\begin{theorem}\label{th4.4}
Let $\mathcal{C}$ be
a $\lambda$-constacyclic code over $R$ of length $np^s$ with canonical form decomposition $\mathcal{C}=\bigoplus_{j=1}^r\varepsilon_j(x)C_j$ $({\rm mod} \ x^{p^sn}-\lambda)$, where $C_j$ is an ideal of $\mathcal{K}_j+u\mathcal{K}_j$ listed by Theorem \ref{th3.8}. Then the dual code
$\mathcal{C}^{\bot}$ of $\mathcal{C}$ is a $\lambda^{-1}$-constacyclic code over $R$ of length $np^s$ with canonical form decomposition
\begin{center}
$\mathcal{C}^{\bot}=\bigoplus_{j=1}^r\widehat{\varepsilon}_j(x)\widehat{D}_j$ $({\rm mod} \ x^{np^s}-\lambda^{-1})$,
\end{center}
where $\widehat{D}_j$ is an ideal of $\widehat{\mathcal{R}}_j+u\widehat{\mathcal{R}}_j$ given by one of the following five cases:

\vskip 2mm\par
  (I) $\widehat{D}_j=\left\langle -\lambda x^{p^sn-d_j}\widetilde{f}_j(x)b(x^{-1})+u\right\rangle$, if $C_j=\langle f_j(x)b(x)+u\rangle$
  where
$b(x)\in f_j(x)^{\lceil \frac{p^s}{2}\rceil-1}(\mathcal{K}_j/\langle f_j(x)^{p^s-1}\rangle)$.

\vskip 2mm\par
  (II) $\widehat{D}_j=\left\langle -\lambda x^{p^sn-d_j}\widetilde{f}_j(x)b(x^{-1})+u, \widetilde{f}_j(x)^{p^s-k}\right\rangle$, if $C_j=\langle f_j(x)^{k+1}b(x)+uf_j(x)^k\rangle$ where $b(x)\in f_j(x)^{\lceil \frac{p^s-k}{2}\rceil-1}(\mathcal{K}_j/\langle f_j(x)^{p^s-k-1}\rangle)$ and $1\leq k\leq p^s-1$.

\vskip 2mm\par
  (III) $\widehat{D}_j=\left\langle \widetilde{f}_j(x)^{p^s-k}\right\rangle$, if $C_j=\langle f_j(x)^k\rangle$ where $0\leq k\leq p^s$.

\vskip 2mm\par
  (IV) $\widehat{D}_j=\left\langle -\lambda x^{p^sn-d_j}\widetilde{ f}_j(x)^{p^s-t+1}b(x^{-1})+u\widetilde{f}_j(x)^{p^s-t}\right\rangle$, if $C_j=\langle f_j(x)b(x)$ $+u,f_j(x)^{t}\rangle$ where  $b(x)\in f_j(x)^{\lceil\frac{t}{2}\rceil-1}(\mathcal{K}_j/\langle f_j(x)^{t-1}\rangle)$,
$1\leq t\leq p^s-1$.

\vskip 2mm\par
  (V) $\widehat{D}_j=\left\langle-\lambda x^{p^sn-d_j}\widetilde{f}_j(x)^{p^s-k-t+1}b(x^{-1})+u\widetilde{f}_j(x)^{p^s-k-t},\widetilde{f}_j(x)^{p^s-k}\right\rangle$,
  if $C_j=\langle f_j(x)^{k+1}b(x)+uf_j(x)^k,f_j(x)^{k+t}\rangle$
where
$b(x)\in f_j(x)^{\lceil\frac{t}{2}\rceil-1}(\mathcal{K}_j/\langle f_j(x)^{t-1}\rangle)$,
$1\leq t\leq p^s-k-1$ and $1\leq k\leq p^s-2$.
\end{theorem}

  \begin {proof} For each integer $j$, $1\leq j\leq r$, let $B_j$ be an ideal of $\mathcal{K}_j+\mathcal{K}_j$ given
by one of the following five cases:

\vskip 2mm \par
  (i) $B_j=\langle -f_j(x)b(x)+u\rangle$, if $C_j=\langle f_j(x)b(x)+u\rangle$ is given by Theorem \ref{th3.8}(I).

\vskip 2mm \par
  (ii) $B_j=\langle -f_j(x)b(x)+u, f_j(x)^{p^s-k}\rangle$, if $C_j=\langle f_j(x)^{k+1}b(x)+uf_j(x)^k\rangle$ is given by Theorem \ref{th3.8}(II).

\vskip 2mm \par
  (iii) $B_j=\langle f_j(x)^{p^s-k}\rangle$, if $C_j=\langle f_j(x)^k\rangle$ is given by Theorem \ref{th3.8}(III).

\vskip 2mm \par
  (iv) $B_j=\langle -f_j(x)^{p^s-t+1}b(x)+uf_j(x)^{p^s-t}\rangle$, if $C_j=\langle f_j(x)b(x)+u,f_j(x)^{t}\rangle$ is given by Theorem \ref{th3.8}(IV).

\vskip 2mm \par
  (v) $B_j=\langle-f_j(x)^{p^s-k-t+1}b(x)+uf_j(x)^{p^s-k-t},f_j(x)^{p^s-k}\rangle$,
  if $C_j=\langle f_j(x)^{k+t}$, $f_j(x)^{k+1}b(x)+uf_j(x)^k\rangle$ is given by Theorem \ref{th3.8}(V).

\vskip 2mm \par
  By Theorem \ref{th3.8}, one can easily verify that
\begin{equation}
\label{eq7}C_j\cdot B_j=\{0\} \ {\rm and} \ |C_j||B_j|=p^{2mp^sd_j}.
\end{equation}
   Now, we denote $\widehat{D}_j=\tau_j(B_j)$ which is an ideal of the ring   $\widehat{\mathcal{K}}_j+u\widehat{\mathcal{K}}_j$, and
set $\mathcal{D}=\sum_{j=1}^r\widehat{\varepsilon}_j(x)\widehat{D}_j=\widehat{\Phi}(\widehat{D}_1\times\ldots\times \widehat{D}_r)$. Then $\mathcal{D}$ is an
ideal of the ring $\widehat{\mathcal{A}}+u\widehat{\mathcal{A}}$, i.e. a $\lambda^{-1}$-constacyclic code over $R$ of
length $np^s$, and $\mathcal{D}=\sum_{j=1}^r\tau(\varepsilon_j(x))\tau_j(B_j)=\sum_{j=1}^r\tau(\varepsilon_j(x)B_j)=\tau(\sum_{j=1}^r\varepsilon_j(x)B_j)$
by Lemma \ref{le4.2}. From this, by Lemma \ref{le2.3}(i) and Equation (\ref{eq7}) we deduce
\begin{eqnarray*}
\tau(\mathcal{C})\cdot \mathcal{D} &=& \tau\left(\sum_{j=1}^r\varepsilon_j(x)C_j\right)\cdot \mathcal{D}=\tau\left(\left(\sum_{j=1}^r\varepsilon_j(x)C_j\right)\left(\sum_{j=1}^r\varepsilon_j(x)B_j\right)\right)\\
  &=& \tau\left(\sum_{j=1}^r\varepsilon_j(x)(C_j\cdot B_j)\right)=\{0\}.
\end{eqnarray*}
This implies that $[\xi,\eta]=0$ for all $\xi\in \mathcal{C}$ and $\eta\in \mathcal{D}$ by Lemma \ref{le4.3}. Hence
$\mathcal{D}\subseteq \mathcal{C}^{\bot}$. Moreover, by Corollary \ref{co3.9} and Equation (\ref{eq7}) it follows that
\begin{eqnarray*}
|\mathcal{C}||\mathcal{D}|&=&(\prod_{j=1}^r|C_j|)(\prod_{j=1}^r|B_j|)
=\prod_{j=1}^r|C_j||B_j|=p^{2mp^s\sum_{j=1}^rd_j}\\
 &=&p^{2mp^sn}=|\mathbb{F}_{p^m}+u\mathbb{F}_{p^m}|^{np^s}=|R|^{np^s}.
\end{eqnarray*}
Since $R$ is a Frobenius ring, we conclude that $\mathcal{D}=\mathcal{C}^{\bot}$ (see \cite{s15}).

\par
  Finally, we give the explicit representation for each $\mathcal{C}^{\bot}$. By
$\lambda x^{np^s}=1$ in $\mathcal{R}_{\lambda^{-1}}$ and ${\rm deg}(f_j(x))=d_j<n$, for any integer $l$, $1\leq l\leq p^s-1$, we have
\begin{equation}\label{eq8}\tau_j(f_j(x)^l)=f_j(x^{-1})^l=\lambda x^{np^s-ld_j}(x^{d_j}f_j(x^{-1}))^l=\lambda x^{np^s-ld_j}\widetilde{f}_j(x)^l.
\end{equation}

\par
   $\diamond$ When $B_j$ is given in Case (i), by Equation (\ref{eq8}) we have
\begin{eqnarray*}
\widehat{D}_j&=&\tau_j(B_j)=\langle \tau_j(-f_j(x)b(x)+u)\rangle=\langle -\tau_j(f_j(x))\tau_j(b(x))+u\rangle\\
 &=&\langle -\lambda x^{np^s-d_j}\widetilde{f}_j(x)b(x^{-1})+u\rangle
\end{eqnarray*}

\par
   $\diamond$ When $B_j$ is given in Case (ii), by Equation (\ref{eq8}) we have
\begin{eqnarray*}
\widehat{D}_j&=&\tau_j(B_j)=\langle \tau_j(-f_j(x)b(x)+u), \tau_j(f_j(x)^{p^s-k}) \rangle\\
  &=&\langle -\lambda x^{np^s-d_j}\widetilde{f}_j(x)b(x^{-1})+u, \lambda x^{np^s-(p^s-k)d_j}\widetilde{f}_j(x)^{p^s-k}\rangle\\
  &=& \langle -\lambda x^{np^s-d_j}\widetilde{f}_j(x)b(x^{-1})+u, \widetilde{f}_j(x)^{p^s-k}\rangle.
\end{eqnarray*}

\par
   $\diamond$ When $B_j$ is given in Case (iii), it can be proved similarly as that of Case (ii).

\par
   $\diamond$ When $B_j$ is given in Case (iv), by Equation (\ref{eq8}) we have
\begin{eqnarray*}
\widehat{D}_j&=&\tau_j(B_j)=\langle \tau_j(-f_j(x)^{p^s-t+1}b(x)+uf_j(x)^{p^s-t}) \rangle\\
  &=&\langle -\lambda x^{np^s-(p^s-t+1)d_j}\widetilde{f}_j(x)^{p^s-t+1}b(x^{-1})+u \cdot \lambda x^{np^s-(p^s-t)d_j}\widetilde{f}_j(x)^{p^s-t}   \rangle\\
  &=& \langle -x^{-d_j}\widetilde{f}_j(x)^{p^s-t+1}b(x^{-1})+u \widetilde{f}_j(x)^{p^s-t}\rangle,
\end{eqnarray*}
where $x^{-d_j}=\lambda x^{np^s-d_j}$ in $\widehat{\mathcal{K}}_j=\mathbb{F}_{p^m}[x]/\langle \widetilde{f}_j(x)^{p^s}\rangle$ as
$\widetilde{f}_j(x)^{p^s}|(x^{np^s}-\lambda^{-1})$.

\par
   $\diamond$ When $B_j$ is given in Case (v), it can be proved similarly as that of Cases (ii) and (iv).
\end{proof}

\begin{example} By \cite{s21} Example 10.1, $x^{2^i}-3$ is irreducible over $\mathbb{F}_5$ for any
integer $i\geq 0$. It is clear that $3^5=3$ and $x^{2^i}-3=x^{2^i}+2$ in $\mathbb{F}_5[x]$. Now, we consider
the special case of $i=2$. By $p=5$, $s=1$ and $n=4$, we have

\vskip 2mm \noindent
  $\diamondsuit$ $\mathcal{T}=\{\sum_{i=0}^3a_ix^i\mid a_0,a_1,a_2,a_3\in \mathbb{F}_5\}$ and $|\mathcal{T}|=5^4$;

\vskip 2mm \noindent
  $\diamondsuit$
    $\mathcal{A}=\mathbb{F}_5[x]/\langle (x^4+2)^5\rangle=\{\sum_{j=0}^4b_j(x)(x^4+2)^j\mid b_j(x)\in \mathcal{T}, \ j=0,1,2,3,4\}$;

\vskip 2mm \noindent
  $\diamondsuit$ $(x^4+2)^l(\mathcal{A}/\langle (x^4+2)^t\rangle)=\{\sum_{j=l}^{t-1}b_j(x)(x^4+2)^j\mid b_j(x)\in \mathcal{T}, \ j=l,\ldots,t-1\}$
for any $0\leq l\leq t\leq 4$, where we set $(x^4+2)^t(\mathcal{A}/\langle (x^4+2)^t\rangle)=\{0\}$.
\end{example}

\par
  By Corollary 3.10, the number of $3$-constacyclic codes over $\mathbb{F}_5+u\mathbb{F}_5$  of length $4\cdot 5=20$ is equal to $N_{(5,4,5)}=3\cdot 5^8+7\cdot 5^4+11=1176261$.

\par
  As $\lambda=3$, $\widetilde{(x^4+2)}=1+2x^4=2(x^4+3)$ and $2\cdot (-\lambda)=4$, by Corollary \ref{co3.10} and Theorem \ref{th4.4} we see that all distinct $3$-constacyclic codes over $\mathbb{F}_5+u\mathbb{F}_5$  of length $20$ and their dual codes are
given by one the following five cases:

\vskip 2mm \par
 (I) $5^{4\cdot 2}=390625$ codes:
$$\mathcal{C}=\langle (x^4+2)b(x)+u\rangle \ {\rm with} \ |\mathcal{C}|=5^{20}$$
and $\mathcal{C}^{\bot}=\langle 4x^{16}(x^4+3)b(x^{-1})+u\rangle$ which is a $2$-constacyclic codes over $\mathbb{F}_5+u\mathbb{F}_5$  of length $20$,
  where
$$b(x)\in (x^4+2)^{2}(\mathcal{A}/\langle (x^4+2)^{4}\rangle)=\{g(x)(x^4+2)^2+h(x)(x^4+2)^3\mid g(x),h(x)\in \mathcal{T}\}.$$

\vskip 2mm\par
  (II) $5^{4\cdot 2}+5^{4\cdot 1}+5^{4\cdot 1}+5^{4\cdot 0}=391876$ codes:
$$\mathcal{C}=\langle (x^4+2)^{k+1}b(x)+u(x^4+2)^k\rangle \ {\rm with} \ |\mathcal{C}|=5^{4(5-k)}$$
and $\mathcal{C}^{\bot}=\langle 4x^{16}(x^4+3)b(x^{-1})+u, (x^4+3)^{5-k}\rangle$,
where
\begin{center}
$b(x)\in (x^4+2)^{\lceil\frac{5-k}{2}\rceil-1}(\mathcal{A}/\langle (x^4+2)^{5-k-1}\rangle)$ and $1\leq k\leq 4$.
\end{center}

\vskip 2mm \par
 (III) $6$ codes:
$\mathcal{C}=\langle (x^4+2)^k\rangle \ {\rm with} \ |\mathcal{C}|=5^{8(5-k)}$
and $\mathcal{C}^{\bot}=\langle (x^4+3)^{5-k}\rangle$, where $0\leq k\leq 5$.

\vskip 2mm \par
 (IV) $5^{4\cdot 0}+5^{4\cdot 1}+5^{4\cdot 1}+5^{4\cdot 2}=391876$ codes:
$$\mathcal{C}=\langle (x^4+2)b(x)+u,(x^4+2)^{t}\rangle \ {\rm with} \ |\mathcal{C}|=5^{4(10-t)}$$
and $\mathcal{C}^{\bot}=\langle 4x^{16}(x^4+3)^{5-t+1}b(x^{-1})+u(x^4+3)^{5-t}\rangle$, where
\begin{center}
$b(x)\in (x^4+2)^{\lceil \frac{t}{2}\rceil-1}(\mathcal{A}/\langle (x^4+2)^{t-1}\rangle)$,
$1\leq t\leq 4$.
\end{center}

\vskip 2mm \par
 (V) $(5^{4\cdot 0}+5^{4\cdot 1}+5^{4\cdot 1})+(5^{4\cdot 0}+5^{4\cdot 1})+5^{4\cdot 0}=1878$ codes:
$$\mathcal{C}=\langle (x^4+2)^{k+1}b(x)+u(x^4+2)^k,(x^4+2)^{k+t}\rangle \ {\rm with} \ |\mathcal{C}|=5^{4(10-2k-t)}$$
and $\mathcal{C}^{\bot}=\langle 4x^{16}(x^4+3)^{5-k-t+1}b(x^{-1})+u(x^4+3)^{p^s-k-t},(x^4+3)^{p^s-k}\rangle$,
where
$b(x)\in (x^4+2)^{\lceil\frac{t}{2}\rceil-1}(\mathcal{A}/\langle (x^4+2)^{t-1}\rangle)$,
$1\leq t\leq 4-k$ and $1\leq k\leq 3$.

\vskip 2mm \par
  The only self-dual $3$-constacyclic codes over $\mathbb{F}_5+u\mathbb{F}_5$  of length $20$
is $\langle u\rangle=u\mathbb{F}_5^{20}$ (corresponding to $b(x)=0$ in Case (I)).


\section{Self-dual cyclic codes and negacyclic codes of length $np^s$ over $R$} \label{}

\noindent \par
In this section, we investigate self-dual cyclic codes and negacyclic codes of length $np^s$ over $R=\mathbb{F}_{p^m}+u\mathbb{F}_{p^m}$.
Let $\nu\in\{1,-1\}$ and $\mathcal{R}_{\nu}=R[x]/\langle x^{np^s}-\nu\rangle$. Then
$\mathcal{C}$ is a \textit{cyclic code} (resp. \textit{negacyclic code}) if and only if $\mathcal{C}$ is an ideal
of the ring $\mathcal{R}_{\nu}$, i.e. a $\nu$-constacyclic code
over $R$ of length $np^s$, where $\nu=1$ (resp. $\nu=-1$).

\par
  Using the notations in Sections 1--4, by $\nu^{-1}=\nu$ we have $\mathcal{R}_{\nu}=\mathcal{R}_{\nu^{-1}}$ and
$\mathcal{A}=\widehat{\mathcal{A}}=\mathbb{F}_{p^m}[x]/\langle x^{p^sn}-\nu\rangle$.
Hence the map $\tau:\mathcal{A}\rightarrow \mathcal{A}$ defined by $\tau(a(x))=a(x^{-1})=a(\nu x^{np^s-1})$ (mod $x^{np^s}-\nu$) for all $a(x)\in \mathcal{A}$
is a ring automorphism on $\mathcal{A}$ satisfying $\tau^{-1}=\tau$.

\par
  By Equations (\ref{eq1}) in Section 2 and (\ref{eq6}) in Section 4, we have
$$x^{np^s}-\nu=f_1(x)^{p^s}\ldots f_r(x)^{p^s}, \
x^{np^s}-\nu=-\nu\widetilde{(x^{p^sn}-\nu})=-\nu\widetilde{f}_1(x)^{p^s}\ldots \widetilde{f}_r(x)^{p^s}.$$
Since $f_1(x),\ldots, f_r(x)$ are pairwise coprime monic irreducible polynomials in $\mathbb{F}_{p^m}[x]$, for
each $1\leq j\leq r$ there is a unique integer $j^{\prime}$, $1\leq j^{\prime}\leq r$, such that
$\widetilde{f}_j(x)=\delta_jf_{j^{\prime}}(x)$ where $\delta_j=f_j(0)^{-1}\in \mathbb{F}_{p^m}^{\times}$.
In the following, we still denote the bijection $j\mapsto j^{\prime}$ on $\{1,\ldots,r\}$
by $\tau$, i.e.,
$$\widetilde{f}_j(x)=\delta_jf_{\tau(j)}(x).$$
   Whether $\tau$ denotes the automorphism of ${\mathcal A}$ or this map on $\{1,2$, $\ldots,r\}$ is determined by context.
The next lemma shows the compatibility of the two uses of $\tau$.

\begin{lemma}\label{le5.1}
Using the notations above, we have the following conclusions.

\vskip 2mm \par
   (i) \textit{$\tau$ is a permutation on the set $\{1,\ldots,r\}$ satisfying $\tau^{-1}=\tau$}.

\vskip 2mm \par
   (ii) \textit{After a suitable rearrangement of $f_1(x),\ldots,f_r(x)$, there are nonnegative integers $\rho,\epsilon$ such that
$\rho+2\epsilon=r$, $\tau(j)=j$ for all $j=1,\ldots,\rho$,  $\tau(\rho+i)=\rho+\epsilon+i$  and
$\tau(\rho+\epsilon+i)=\rho+i$ for all $i=1,\ldots,\epsilon$}.

\vskip 2mm \par
   (iii) \textit{For any $1\leq j\leq r$, $\tau(\varepsilon_j(x))=\widehat{\varepsilon}_{j}(x)=\varepsilon_j(x^{-1})=\varepsilon_{\tau(j)}(x)$ in the ring ${\mathcal A}$}.

\vskip 2mm \par
   (iv) \textit{For any integer $j$, $1\leq j\leq r$, the map
$\tau_j:\mathcal{K}_j+u\mathcal{K}_j\rightarrow \mathcal{K}_{\tau(j)}+u\mathcal{K}_{\tau(j)}$ defined
by
$\tau_j(a(x)+ub(x))=a(x^{-1})+ub(x^{-1}) \ (\forall a(x),b(x)\in \mathcal{K}_j)$
is a ring isomorphism
from $\mathcal{K}_j+u\mathcal{K}_j$ onto $\mathcal{K}_{\tau(j)}+u\mathcal{K}_{\tau(j)}$ with inverse
$\tau_j^{-1}=\tau_{\tau(j)}$}.

\par
   Moreover,
  we have $\tau(\varepsilon_j(x)\xi)=\varepsilon_{\tau(j)}(x)\tau_j(\xi)$ for all $\xi\in \mathcal{K}_j+u\mathcal{K}_j$.
\end{lemma}

\begin{proof} (i) By $\widetilde{f}_j(x)=\delta_jf_{\tau(j)}(x)$ for all $j$, it follows that
$f_{\tau(j)}(x)=\delta_j^{-1}\widetilde{f}_j(x)$. This implies
$\widetilde{f_{\tau(j)}}(x)=\delta_j^{-1}\widetilde{\widetilde{f}_j(x)}=\delta_j^{-1}f_j(x)$, and
so $\delta_{\tau(j)}=\delta_j^{-1}$.
Hence
$f_{\tau^2(j)}(x)=f_{\tau(\tau(j))}(x)=\delta_{\tau(j)}^{-1}\widetilde{f}_{\tau(j)}(x)
=\delta_{\tau(j)}^{-1}\delta_j^{-1}\widetilde{\widetilde{f}_j(x)}=f_j(x)$. This implies
$\tau^2(j)=j$ for all $j\in\{1,\ldots,r\}$. Therefore, $\tau^{-1}=\tau$.

\par
  (ii) It follows from (i) immediately.

\par
  (iii) From the equation $v_j(x)F_j(x)+w_j(x)f_j(x)=1$ in Section 2, we deduce that
$v_j(x^{-1})F_j(x^{-1})+w_j(x^{-1})f_j(x^{-1})=\tau(v_j(x)F_j(x)+w_j(x)f_j(x))=1$ in $\mathcal{A}$. By
Equation (\ref{eq8}), we have $f_j(x^{-1})^{p^s}=\nu x^{np^s-p^sd_j}\widetilde{f}_j(x)^{p^s}=\nu \delta_j^{p^s}x^{p^s(n-d_j)}f_{\tau(j)}(x)^{p^s}$. From this, by $F_j(x)=\frac{x^n-\nu}{f_j(x)}$ and $\nu^{p^s}=\nu$ we deduce that
\begin{eqnarray*}
(F_j(x^{-1}))^{p^s}&=&\left(\frac{x^{-n}-\nu}{f_j(x^{-1})}\right)^{p^s}
 =\frac{x^{-np^s}-\nu}{f_j(x^{-1})^{p^s}}
 =\frac{-\nu x^{-np^s}(x^{np^s}-\nu)}{\nu \delta_j^{p^s}x^{p^s(n-d_j)}f_{\tau(j)}(x)^{p^s}}\\
 &=&-\nu\delta_j^{-p^s}x^{p^s(d_j-n)}F_{\tau(j)}(x)^{p^s}.
\end{eqnarray*}
Hence $\left(-\nu\delta_j^{-p^s}x^{p^s(d_j-n)}v_j(x^{-1})^{p^s}\right)F_{\tau(j)}(x)^{p^s}+\left(\nu x^{np^s-p^sd_j}w(x^{-1})^{p^s}\right)f_{\tau(j)}(x)^{p^s}
=\left(v_j(x^{-1})F_j(x^{-1})+w_j(x^{-1})f_j(x^{-1})\right)^{p^s}=1$. Then by Equation (\ref{eq3}) we have
\begin{eqnarray*}
\tau(\varepsilon_j(x))&=&\varepsilon_j(x^{-1})=v_j(x^{-1})^{p^s}F_j(x^{-1})^{p^s}\\
 &=&\left(-\nu\delta_j^{-p^s}x^{p^s(d_j-n)}v_j(x^{-1})^{p^s}\right)F_{\tau(j)}(x)^{p^s}\\
 &=&1-\left(\nu x^{np^s-p^sd_j}w(x^{-1})^{p^s}\right)f_{\tau(j)}(x)^{p^s}
\end{eqnarray*}
in $\mathcal{A}$. This implies $\tau(\varepsilon_j(x))=\varepsilon_{\tau(j)}(x)$ by
Notation \ref{no2.2}.

\par
  (iv) By Notation \ref{no4.1} and $\widetilde{f}_j(x)=\delta_jf_{\tau(j)}(x)$, we have
\begin{eqnarray*}
\widehat{\mathcal{K}}_j=\mathbb{F}_{p^m}[x]/\langle \widetilde{f}_j(x)^{p^s}\rangle=\mathbb{F}_{p^m}[x]/\langle \delta_j^{p^s}f_{\tau(j)}(x)^{p^s}\rangle
=\mathbb{F}_{p^m}[x]/\langle f_{\tau(j)}(x)^{p^s}\rangle=\mathcal{K}_{\tau(j)}.
\end{eqnarray*}
Then by Lemma \ref{le4.2} we deduce that
$\tau_j$
is a ring isomorphism
from $\mathcal{K}_j+u\mathcal{K}_j$ onto $\mathcal{K}_{\tau(j)}+u\mathcal{K}_{\tau(j)}$ and
$\tau(\varepsilon_j(x)\xi)=\varepsilon_{\tau(j)}(x)\tau_j(\xi)$ for all $\xi\in \mathcal{K}_j+u\mathcal{K}_j$.
\end{proof}

\vskip 3mm\par
   Then by Lemma \ref{le5.1} and Theorem \ref{th4.4}, we have the following conclusion for dual codes of cyclic and negacyclic codes
over $R$ of length $np^s$.

\vskip 3mm\noindent
\begin{corollary}\label{co5.2} Let $\nu\in\{1,-1\}$ and $\mathcal{C}$ be a $\nu$-constacyclic code over $R$ of length $np^s$ with canonical form decomposition $\mathcal{C}=\bigoplus_{j=1}^r\varepsilon_j(x)C_j$ $({\rm mod} \ x^{p^sn}-\nu)$, where $C_j$ is an ideal of $\mathcal{K}_j+u\mathcal{K}_j$ listed by Theorem \ref{th3.8}. Then the dual code $\mathcal{C}^{\bot}$ of $\mathcal{C}$ is a $\nu$-constacyclic code over $R$ of length $np^s$ with canonical form decomposition
\begin{center}
$\mathcal{C}^{\bot}=\bigoplus_{j=1}^r\varepsilon_{\tau(j)}(x)D_{\tau(j)}$ $({\rm mod} \ x^{np^s}-\nu)$,
\end{center}
where $D_{\tau(j)}$ is an ideal of $\mathcal{K}_{\tau(j)}+u\mathcal{K}_{\tau(j)}$ given by one of the following five cases:

\vskip 2mm\par
  (I) If $C_j=\langle f_j(x)b(x)+u\rangle$
  where
$b(x)\in f_j(x)^{\lceil \frac{p^s}{2}\rceil-1}(\mathcal{K}_j/\langle f_j(x)^{p^s-1}\rangle)$,
\begin{center}
$D_{\tau(j)}=\langle -\nu \delta_jx^{p^sn-d_j}f_{\tau(j)}(x)b(x^{-1})+u\rangle$.
\end{center}

\vskip 2mm\par
  (II) If $C_j=\langle f_j(x)^{k+1}b(x)+uf_j(x)^k\rangle$ where $1\leq k\leq p^s-1$ and $b(x)\in f_j(x)^{\lceil \frac{p^s-k}{2}\rceil-1}(\mathcal{K}_j/\langle f_j(x)^{p^s-k-1}\rangle)$,
\begin{center}
$D_{\tau(j)}=\langle -\nu\delta_j x^{p^sn-d_j}f_{\tau(j)}(x)b(x^{-1})+u, f_{\tau(j)}(x)^{p^s-k}\rangle$.
\end{center}

\vskip 2mm\par
  (III) If $C_j=\langle f_j(x)^k\rangle$ where $0\leq k\leq p^s$, $D_{\tau(j)}=\langle f_{\tau(j)}(x)^{p^s-k}\rangle$.

\vskip 2mm\par
  (IV) If $C_j=\langle f_j(x)b(x)+u,f_j(x)^{t}\rangle$ where $1\leq t\leq p^s-1$ \\
  and $b(x)\in f_j(x)^{\lceil\frac{t}{2}\rceil-1}(\mathcal{K}_j/\langle f_j(x)^{t-1}\rangle)$,
\begin{center}
$D_{\tau(j)}=\langle -\nu\delta_j x^{p^sn-d_j}f_{\tau(j)}(x)^{p^s-t+1}b(x^{-1})+uf_{\tau(j)}(x)^{p^s-t}\rangle$.
\end{center}

\vskip 2mm\par
  (V) If $C_j=\langle f_j(x)^{k+1}b(x)+uf_j(x)^k,f_j(x)^{k+t}\rangle$
where $1\leq t\leq p^s-k-1$, $1\leq k\leq p^s-2$,
 and $b(x)\in f_j(x)^{\lceil\frac{t}{2}\rceil-1}(\mathcal{K}_j/\langle f_j(x)^{t-1}\rangle)$,
\begin{center}
$D_{\mu(j)}=\langle -\nu\delta_j x^{p^sn-d_j}f_{\tau(j)}(x)^{p^s-k-t+1}b(x^{-1})+uf_{\tau(j)}(x)^{p^s-k-t},f_{\tau(j)}(x)^{p^s-k}\rangle$.
\end{center}
\end{corollary}

   Now, we can determine self-dual cyclic and negacyclic codes of length $p^sn$ over $R$ by the following theorem.

\begin{theorem}\label{th5.3}
Using the notations in Lemma \ref{le5.1}$({\rm ii})$ and $x^{-1}=\nu x^{np^s-1}$, all distinct
self-dual $\nu$-constacyclic codes of length $p^sn$ over $R$ are given by
$$\mathcal{C}=\bigoplus_{j=1}^r\varepsilon_j(x)C_j \ ({\rm mod} \ x^{np^s}-\nu),$$
where $C_j$ is an ideal of $\mathcal{K}_j+u\mathcal{K}_j$ given by one of the following two cases:

\vskip 2mm\par
  (i) If $1\leq j\leq \rho$, $C_j$ is given by one of the following three subcases.

\vskip 2mm \par
   (i-1) $C_j=\langle f_j(x)b(x)+u\rangle$, where
$b(x)\in f_j(x)^{\lceil \frac{p^s}{2}\rceil-1}(\mathcal{K}_j/\langle f_j(x)^{p^s-1}\rangle)$ satisfying
$b(x)+\nu\delta_jx^{p^sn-d_j}b(x^{-1})\equiv 0 \ ({\rm mod} \ f_j(x)^{p^s-1})$.

\vskip 2mm \par
   (i-2) $C_j=\langle f_j(x)^k\rangle$ where $k$ is an integer satisfying $2k=p^s$.

\vskip 2mm \par
   (i-3) $C_j=\langle f_j(x)^{k+1}b(x)+uf_j(x)^k,f_j(x)^{k+t}\rangle$, where
$1\leq t\leq p^s-k-1$, $1\leq k\leq p^s-2$,
 and $b(x)\in f_j(x)^{\lceil\frac{t}{2}\rceil-1}(\mathcal{K}_j/\langle f_j(x)^{t-1}\rangle)$ satisfying
$p^s=2k+t$ and
$b(x)+\nu\delta_jx^{p^sn-d_j}b(x^{-1})\equiv 0 \ ({\rm mod} \ f_j(x)^{t-1})$.

\vskip 2mm\par
  (ii) If $j=\rho+i$ where $1\leq i\leq \epsilon$, the pair $(C_j,C_{j+\epsilon})$ of ideals is given by one of the following five subcases.

\vskip 2mm\par
  (ii-1) $C_j=\langle f_j(x)b(x)+u\rangle$ and $C_{j+\epsilon}=\langle -\nu\delta_jx^{p^sn-d_j}f_{j+\epsilon}(x)b(x^{-1})+u\rangle$
  where
$b(x)\in f_j(x)^{\lceil \frac{p^s}{2}\rceil-1}(\mathcal{K}_j/\langle f_j(x)^{p^s-1}\rangle)$.

\vskip 2mm\par
  (ii-2) $C_j=\langle f_j(x)^{k+1}b(x)+uf_j(x)^k\rangle$ and
\begin{center}
$C_{j+\epsilon}=\langle -\nu\delta_j x^{p^sn-d_j}f_{j+\epsilon}(x)b(x^{-1})+u, f_{j+\epsilon}(x)^{p^s-k}\rangle$,
\end{center}
where $b(x)\in f_j(x)^{\lceil \frac{p^s-k}{2}\rceil-1}(\mathcal{K}_j/\langle f_j(x)^{p^s-k-1}\rangle)$ and $1\leq k\leq p^s-1$.

\vskip 2mm\par
  (ii-3) $C_j=\langle f_j(x)^k\rangle$ and $C_{j+\epsilon}=\langle f_{j+\epsilon}(x)^{p^s-k}\rangle$, where $0\leq k\leq p^s$.

\vskip 2mm\par
  (ii-4) $C_j=\langle f_j(x)b(x)+u,f_j(x)^{t}\rangle$ and
\begin{center}
$C_{j+\epsilon}=\langle -\nu\delta_j x^{p^sn-d_j}f_{j+\epsilon}(x)^{p^s-t+1}b(x^{-1})+uf_{j+\epsilon}(x)^{p^s-t}\rangle$,
\end{center}
where $b(x)\in f_j(x)^{\lceil\frac{t}{2}\rceil-1}(\mathcal{K}_j/\langle f_j(x)^{t-1}\rangle)$ and $1\leq t\leq p^s-1$.

\vskip 2mm\par
  (ii-5) $C_j=\langle f_j(x)^{k+1}b(x)+uf_j(x)^k,f_j(x)^{k+t}\rangle$ and
\begin{center}
$C_{j+\epsilon}=\langle -\nu\delta_j x^{p^sn-d_j}f_{j+\epsilon}(x)^{p^s-k-t+1}b(x^{-1})+uf_{j+\epsilon}(x)^{p^s-k-t},f_{j+\epsilon}(x)^{p^s-k}\rangle$,
\end{center}
where $b(x)\in f_j(x)^{\lceil\frac{t}{2}\rceil-1}(\mathcal{K}_j/\langle f_j(x)^{t-1}\rangle)$, $1\leq t\leq p^s-k-1$ and $1\leq k\leq p^s-2$.
\end{theorem}

\begin{proof}
We have $\mathcal{C}^{\bot}=\bigoplus_{j=1}^r\varepsilon_{\tau(j)}(x)D_{\tau(j)}$,
where $D_{\tau(j)}$ is an ideal of $\mathcal{K}_{\tau(j)}+u\mathcal{K}_{\tau(j)}$ given by Corollary \ref{co5.2} for $j=1,\ldots,r$.
Since $\tau$ is a bijection on the set $\{1,\ldots,r\}$, we have $\mathcal{C}=\bigoplus_{j=1}^r\varepsilon_{\tau(j)}(x)C_{\tau(j)}$.
From this and by
Theorem \ref{th2.5}(iii), we deduce that $\mathcal{C}$ is self-dual if and only if $C_{\tau(j)}=D_{\tau(j)}$ for all $j=1,\ldots,r$.
Then by Lemma \ref{le5.1}(ii), we have one of the following two cases.

\par
   (i) Let $1\leq j\leq \rho$. Then $\tau(j)=j$. By Corollary \ref{co5.2}, $C_j$ satisfies the condition $C_{\tau(j)}=D_{\tau(j)}$ if and only
if $C_j$ is given one of the following subcases:

\par
   (i-1) $C_j=\langle f_j(x)b(x)+u\rangle=\langle  -\nu\delta_jx^{p^sn-d_j}f_{\tau(j)}(x)b(x^{-1})+u\rangle$,
    where $b(x)\in f_j(x)^{\lceil \frac{p^s}{2}\rceil-1}(\mathcal{K}_j/\langle f_j(x)^{p^s-1}\rangle)$ satisfying $b(x)f_j(x)=-\nu\delta_jx^{p^sn-d_j}b(x^{-1})f_j(x)$ in $\mathcal{K}_j$, i.e.
$b(x)+\nu\delta_jx^{p^sn-d_j}b(x^{-1})\equiv 0 \ ({\rm mod} \ f_j(x)^{p^s-1}).$

\par
   (i-2) $C_j=\langle f_j(x)^k\rangle=\langle f_{\tau(j)}(x)^{p^s-k}\rangle$, where $0\leq k\leq p^s$ satisfying
$p^s-k=k$, i.e., $2k=p^s$.

\par
   (i-3) $C_j=\langle f_j(x)^{k+1}b(x)+uf_j(x)^k,f_j(x)^{k+t}\rangle$ and
\begin{center}
$C_j=\langle -\nu\delta_j x^{p^sn-d_j}f_{\tau(j)}(x)^{p^s-k-t+1}b(x^{-1})+uf_{\tau(j)}(x)^{p^s-k-t},f_{\tau(j)}(x)^{p^s-k}\rangle$,
\end{center}
where $1\leq t\leq p^s-k-1$, $1\leq k\leq p^s-2$,
 and $b(x)\in f_j(x)^{\lceil\frac{t}{2}\rceil-1}(\mathcal{K}_j/\langle f_j(x)^{t-1}\rangle)$ satisfying
$k+t=p^s-k$, i.e. $p^s=2k+t$, and
$b(x)+\nu\delta_jx^{p^sn-d_j}b(x^{-1})\equiv 0 \ ({\rm mod} \ f_j(x)^{t-1}).$

\par
   (ii) Let $j=\rho+i$ where $1\leq i\leq \epsilon$. Then $\tau(j)=j+\epsilon$. In this case,
we choose an arbitrary ideal $C_j$ of $\mathcal{K}_j+u\mathcal{K}_j$ listed in Theorem \ref{th3.8} and
let $C_{j+\epsilon}=C_{\tau(j)}=D_{\tau(j)}=\tau_j(B_j)$, where $D_{\tau(j)}$ is given by Corollary \ref{co5.2} and
$B_j$ is given by the proof of Theorem \ref{th4.4} with $\tau(j)=j+\epsilon$, respectively.
Then the condition $C_{\tau(j)}=D_{\tau(j)}$ is satisfied for for all $j=\rho+1,\ldots,\rho+\epsilon$.

\par
  Moreover, by Equation (\ref{eq7}) in the proof
of Theorem \ref{th4.4} we have $C_j\cdot B_j=\{0\}$ and $|C_j||B_j|=p^{2mp^sd_j}$, where $B_j$ is an ideal
of $\mathcal{K}_j+u\mathcal{K}_j$. By Lemma \ref{le5.1}(ii) and (iv), we know that $\tau_j$ is a ring isomorphism form $\mathcal{K}_j+u\mathcal{K}_j$
onto $\mathcal{K}_{\tau(j)}+u\mathcal{K}_{\tau(j)}=\mathcal{K}_{j+\epsilon}+u\mathcal{K}_{j+\epsilon}$ with
inverse $\tau_j^{-1}=\tau_{\tau(j)}=\tau_{j+\epsilon}$.
This implies that $\tau_j(B_j)$ and $\tau_j(C_j)$ are ideals of $\mathcal{K}_{j+\epsilon}+u\mathcal{K}_{j+\epsilon}$
satisfying
\begin{center}
$\tau_j(B_j)\cdot \tau_j(C_j)=\{0\}$ and $|\tau_j(B_j)||\tau_j(C_j)|=|B_j||C_j|=p^{2mp^sd_j}=p^{2mp^sd_{j+\epsilon}}$.
\end{center}
From this, by $C_{j+\epsilon}=\tau_j(B_j)$ and the proof of Theorem \ref{th4.4} we deduce that $B_{j+\epsilon}=\tau_j(C_j)$. Hence
$D_{\tau(j+\epsilon)}=\tau_{j+\epsilon}(B_{j+\epsilon})
=\tau_{j+\epsilon}(\tau_j(C_j))=C_j=C_{\tau(j+\epsilon)}$
for all $j=\rho+1,\ldots,\rho+\epsilon$.

\par
  As stated above, we conclude that the condition $C_{\tau(j)}=D_{\tau(j)}$ is satisfied for all $j=\rho+1,\ldots,\rho+\epsilon,
\rho+\epsilon+1,\ldots,\rho+2\epsilon$.
\end{proof}

\section{Negacyclic codes over $\mathbb{F}_5+u\mathbb{F}_5$ of
length $n 5^s$ where $n=2\cdot 3^t$}

\noindent \par
In this section, let $t$ be an arbitrary positive integer. We consider negacyclic codes of length $2\cdot 3^t\cdot5^s$ over $\mathbb{F}_{5}+u\mathbb{F}_5$. In this case, we have $p=5$, $m=1$,
$n=2\cdot 3^t$ and $\lambda=-1$.
  By \cite{s3} Section 4, we know that
$$x^{2\cdot 3^t}+1=\prod_{i=1}^{t+1}f_i(x)f_{t+1+i}(x)$$
is the factorization of $x^{2\cdot 3^t}+1$ into monic irreducible factors in $\mathbb{F}_5[x]$, where

\vskip 2mm\par
  $f_1(x)=x+2$ with degree $d_1={\rm deg}(f_1(x))=1$,  $f_{t+2}(x)=3\widetilde{f}_1(x)=x+3$;

\vskip 2mm\par
  $f_i(x)=x^{2\cdot 3^{i-2}}+2x^{3^{i-2}}+4$  with degree $d_i={\rm deg}(f_i(x))=2\cdot 3^{i-2}$,

\par
  $f_{t+1+i}(x)=4\widetilde{f}_i(x)=x^{2\cdot 3^{i-2}}+3x^{3^{i-2}}+4$, for $i=2,\ldots,t+1$.

\vskip 2mm\noindent
Hence $\rho=0$, $\epsilon=t+1$, $\widetilde{f}_i(x)=\delta_if_{t+1+i}(x)$ where $\delta_1=3$ and $\delta_i=4$ for all $i=2,\ldots,t+1$. Then
$$x^{2\cdot5^s\cdot 3^t}+1=(x^{2\cdot 3^t}+1)^{5^s}=\prod_{i=1}^{t+1}f_i(x)^{5^s}f_{t+1+i}(x)^{5^s}.$$

\noindent
   For any integer $i$, $1\leq i\leq t+1$, we find polynomials $a_i(x),b_i(x)\in \mathbb{F}_5[x]$ satisfying
$a_i(x)\frac{x^{2\cdot 3^t}+1}{f_i(x)}+b_i(x)f_i(x)=1$. Then
set
$$\varepsilon_i(x)\equiv a_i(x^{5^s})\frac{x^{2\cdot 3^t\cdot 5^s}+1}{f_i(x^{5^s})}=1-b_i(x^{5^s})f_i(x^{5^s}) \ ({\rm mod} \ x^{2\cdot 3^t\cdot 5^s}+1),$$
$$\varepsilon_{t+1+i}(x)=\varepsilon_i(x^{-1})=\varepsilon_i(4x^{2\cdot 3^t\cdot 5^s-1}) \
({\rm mod} \ x^{2\cdot 3^t\cdot 5^s}+1).$$

\par
    Denote $\mathcal{K}_j=\mathbb{F}_5[x]/\langle f_j(x)^{5^s}\rangle$ and $\mathcal{T}_j=\{\sum_{l=0}^{d_j-1}a_lx^l\mid a_0,a_1,\ldots,a_{d_j-1}\in \mathbb{F}_5\}$ for all $j=1,2,\ldots,2(t+1)$. By Theorem \ref{th3.8}, Corollary \ref{co3.9} and Lemma \ref{le5.1},
all distinct negacyclic code over $\mathbb{F}_5+u\mathbb{F}_5$ of length
$2\cdot 3^t\cdot5^s$ are given by:
\begin{equation}\label{eq9}\mathcal{C}=\bigoplus_{j=1}^{2(t+1)}\varepsilon_j(x)C_j,
\end{equation}
where $C_j$ is an ideal of the
ring $\mathcal{K}_j+u\mathcal{K}_j$ $(u^2=0)$, $1\leq j\leq 2(t+1)$, given by one of the following five cases:

\vskip 2mm \par
 (I-$j$) $5^{\left(5^s-\lceil\frac{5^s}{2}\rceil\right)d_j}$ ideals:
$$C_j=\langle f_j(x)b(x)+u\rangle \ {\rm with} \ |C_j|=p^{5^sd_j},$$
where $b(x)\in f_j(x)^{\lceil \frac{5^s}{2}\rceil-1}(\mathcal{K}_j/\langle f_j(x)^{5^s-1}\rangle)$.

\vskip 2mm \par
 (II-$j$) $\sum_{k=1}^{5^s-1}5^{\left(5^s-k-\lceil\frac{1}{2}(5^s-k)\rceil\right)d_j}$ ideals:
$$C_j=\langle f_j(x)^{k+1}b(x)+uf_j(x)^k\rangle \ {\rm with} \ |C_j|=5^{(5^s-k)d_j},$$
where $b(x)\in f_j(x)^{\lceil \frac{5^s-k}{2}\rceil-1}(\mathcal{K}_j/\langle f_j(x)^{5^s-k-1}\rangle)$ and $1\leq k\leq 5^s-1$.

\vskip 2mm \par
 (III-$j$) $5^s+1$ ideals:
$C_j=\langle f_j(x)^k\rangle \ {\rm with} \ |C_j|=p^{2(5^s-k)d_j}, \ 0\leq k\leq 5^s.$

\vskip 2mm \par
 (IV-$j$) $\sum_{t=1}^{5^s-1}5^{\left(t-\lceil\frac{t}{2}\rceil\right)d_j}$ ideals:
$$C_j=\langle f_j(x)b(x)+u,f_j(x)^{t}\rangle \ {\rm with} \ |C_i|=5^{(2\cdot5^s-t)d_j},$$
where
$b(x)\in f_j(x)^{\lceil \frac{t}{2}\rceil-1}(\mathcal{K}_j/\langle f_j(x)^{t-1}\rangle)$,
$1\leq t\leq 5^s-1$.

\vskip 2mm \par
 (V-$j$) $\sum_{k=1}^{5^s-2}\sum_{t=1}^{5^s-k-1}5^{\left(t-\lceil\frac{t}{2}\rceil\right)d_j}$ ideals:
$$C_j=\langle f_j(x)^{k+1}b(x)+uf_j(x)^k,f_j(x)^{k+t}\rangle \ {\rm with} \ |C_j|=5^{(2\cdot 5^s-2k-t)d_j},$$
where
$b(x)\in f_j(x)^{\lceil \frac{t}{2}\rceil-1}(\mathcal{K}_j/\langle f_j(x)^{t-1}\rangle)$,
$1\leq t\leq 5^s-k-1$ and $1\leq k\leq 5^s-2$.

\vskip 2mm \par
   Moreover, the number of codewords contained in the code $\mathcal{C}$ given by Equation (\ref{eq9})
is equal to $\prod_{i=1}^{2(t+1)}|C_j|$, and the number of all negacyclic codes
of length $2\cdot 3^t\cdot 5^s$ over $\mathbb{F}_5+u\mathbb{F}_5$ is equal to
$\prod_{j=1}^{t+1}N_{(5,d_j,5^s)}^2$
where $d_1=1$, $d_j=2\cdot 3^{j-2}$ and
$$N_{(5,d_j,5^{s})}=\sum_{i=0}^{\frac{5^s-1}{2}}(3+4i)5^{(\frac{5^s-1}{2}-i)d_j}, \
j=1,2,\ldots,t+1.$$

\par
   By Theorem \ref{th5.3}, the number of self-dual negacyclic codes
over $\mathbb{F}_5+u\mathbb{F}_5$ of length $2\cdot 3^t\cdot 5^s$ is equal to
$\prod_{j=1}^{t+1}N_{(5,d_j,5^s)}$ and all distinct self-dual negacyclic codes
over $\mathbb{F}_5+u\mathbb{F}_5$ of
length $2\cdot 3^t\cdot 5^s$  are given by the following:
$$\mathcal{C}=\bigoplus_{j=1}^{t+1}\left(\varepsilon_j(x)C_j\oplus \varepsilon_{t+1+j}(x)C_{t+1+j}\right),$$
where $C_j$ is an ideal of the
ring $\mathcal{K}_j+u\mathcal{K}_j$ $(u^2=0)$ given by (I-j)--(V-j) above
  and $C_{t+1+j}$ is an ideal of the
ring $\mathcal{K}_{t+1+j}+u\mathcal{K}_{t+1+j}$ $(u^2=0)$, $1\leq j\leq t$, given by the one of the following five cases:

\vskip 2mm\par
  ($j$-1) If $C_j$ is given in Case (I-$j$), $C_{t+1+j}=\langle \delta_jx^{2\cdot 3^t\cdot 5^s-d_j}f_{t+1+j}(x)b(x^{-1})+u\rangle$.

 \vskip 2mm\par
  ($j$-2) If $C_j$ is given in Case (II-$j$),
$C_{t+1+j}=\langle \delta_j x^{2\cdot 3^t\cdot 5^s-d_j}f_{t+1+j}(x)b(x^{-1})+u, f_{t+1+j}(x)^{5^s-k}\rangle$.

\vskip 2mm\par
  ($j$-3) If $C_j$ is given in Case (III-$j$), $C_{t+1+j}=\langle f_{t+1+j}(x)^{5^s-k}\rangle$.

\vskip 2mm\par
  ($j$-4) If $C_j$ is given in Case (IV-$j$),
$C_{t+1+j}=\langle \delta_j x^{2\cdot 3^t\cdot 5^s-d_j}f_{t+1+j}(x)^{5^s-t+1}b(x^{-1})$ $+uf_{t+1+j}(x)^{5^s-t}\rangle$.

\vskip 2mm\par
  ($j$-5) If $C_j$ is given in Case (V-$j$),
\begin{eqnarray*}
C_{t+1+j}&=&\langle\delta_j x^{2\cdot 3^t\cdot 5^s-d_j}f_{t+1+j}(x)^{5^s-k-t+1}b(x^{-1})+uf_{t+1+j}(x)^{5^s-k-t},\\
         &&f_{t+1+j}(x)^{5^s-k}\rangle.
\end{eqnarray*}

\par
   Especially, we consider negacyclic codes over $\mathbb{F}_5+u\mathbb{F}_5$ of length $6\cdot 5=30$ corresponding to the special case of $s=t=1$.
The number of all negacyclic codes over $\mathbb{F}_5+u\mathbb{F}_5$ of length $30$ is equal to
$N_{(5,1,5)}^2N_{(5,2,5)}^2=121^2\cdot 2061^2=62190883161$, and the number of self-dual negacyclic codes over $\mathbb{F}_5+u\mathbb{F}_5$ of length $30$ is equal to
$N_{(5,1,5)}N_{(5,2,5)}=121\cdot 2061=249381$.

\par
  Obviously, $x^6+1=f_1(x)f_2(x)f_3(x)f_4(x)$ is the factorization of $x^6+1$ into monic irreducible factors in $\mathbb{F}_5[x]$,
where $f_1(x)=x+2$, $f_2(x)=x^2+2x+4$, $f_3(x)=x+3$ and $f_4(x)=x^2+3x+4$ satisfying
$\widetilde{f}_1(x)=\delta_1f_3(x)$ and $\widetilde{f}_2(x)=\delta_1f_4(x)$ with $\delta_1=3$
and $\delta_2=4$. First, we obtain the following:

\vskip 2mm\noindent
   $\varepsilon_1(x)=2x^{25}+ x^{20}+3x^{15}+4x^{10}+2x^5+1$;
   $\varepsilon_2(x)=2x^{25}+4x^{20}+4x^{15}+x^{10}+2x^{5}+2$;

\vskip 2mm\noindent
   $\varepsilon_3(x)=3x^{25}+x^{20}+2x^{15}+4x^{10}+3x^{5}+1$;
   $\varepsilon_4(x)=3x^{25}+4x^{20}+ x^{15}+x^{10}+3x^{5}+2$.

\vskip 2mm\par
   Next, let $\mathcal{K}_j=\mathbb{F}_5[x]/\langle f_j(x)^5\rangle$ for $j=1,2,3,4$. Then all distinct self-dual negacyclic codes over $\mathbb{F}_5+u\mathbb{F}_5$ of length $30$ are given by:
$$\mathcal{C}=\sum_{j=1}^4\varepsilon_j(x)C_j \ ({\rm mod} \ x^{30}+1),$$
where $C_j$ is an ideal of the ring $\mathcal{K}_j+u\mathcal{K}_j$ such that:

\par
  (i) the pair $(C_1,C_3)$ of ideals is given by one of the following five cases.

\vskip 2mm \par
  (i-1) $5^2=25$ pairs:

\par
  $C_1=\langle (x+2)b(x)+u\rangle$ and $C_3=\langle 3x^{29}(x+3)b(x^{-1})+u\rangle$,
where
\begin{center}
$b(x)\in (x+2)^{2}(\mathcal{K}_1/\langle (x+2)^{4}\rangle)=\{g\cdot(x+2)^2+h\cdot (x+2)^3\mid g,h\in \mathbb{F}_5\}$.
\end{center}

\vskip 2mm \par
  (i-2) $5^{2}+5^{1}+5^{1}+5^{0}=36$ pairs:

\par
  $C_1=\langle (x+2)^{k+1}b(x)+u(x+2)^k\rangle$ and $C_{3}=\langle 3 x^{29}(x+3)b(x^{-1})+u, (x+3)^{5-k}\rangle$,
where $b(x)\in (x+2)^{\lceil \frac{5-k}{2}\rceil-1}(\mathcal{K}_1/\langle (x+2)^{5-k-1}\rangle)$ and $1\leq k\leq 4$.

\vskip 2mm \par
  (i-3) $6$ pairs:
  $C_1=\langle (x+2)^k\rangle$ and $C_{3}=\langle (x+3)^{5-k}\rangle$, where $0\leq k\leq 5$.

\vskip 2mm \par
  (i-4) $5^{0}+5^{1}+5^{1}+5^{2}=36$ pairs:

\par
  $C_1=\langle (x+2)b(x)+u,(x+2)^{t}\rangle$ and $C_{3}=\langle 3 x^{29}(x+3)^{5-t+1}b(x^{-1})+u(x+3)^{5-t}\rangle$,
where
$b(x)\in (x+2)^{\lceil \frac{t}{2}\rceil-1}(\mathcal{K}_1/\langle (x+2)^{t-1}\rangle)$,
$1\leq t\leq 4$.

\vskip 2mm \par
  (i-5) $(5^{0}+5^{1}+5^{1})+(5^{0}+5^{1})+5^{0}=18$ pairs:

\par
  $C_1=\langle (x+2)^{k+1}b(x)+u(x+2)^k,(x+2)^{k+t}\rangle$ and
\begin{center}
  $C_{3}=\langle 3 x^{29}(x+3)^{5-k-t+1}b(x^{-1})+u(x+3)^{5-k-t},(x+3)^{5-k}\rangle$,
\end{center}
where
$b(x)\in (x+2)^{\lceil \frac{t}{2}\rceil-1}(\mathcal{K}_1/\langle (x+2)^{t-1}\rangle)$,
$1\leq t\leq 5-k-1$ and $1\leq k\leq 3$.

\vskip 2mm \par
  (ii) $(C_2,C_4)$ is given by one of the following five cases.

\vskip 2mm \par
  (ii-1) $5^{2\cdot 2}=625$ pairs:

\par
  $C_2=\langle (x^2+2x+4)b(x)+u\rangle$ and $C_4=\langle 3x^{28}(x^2+3x+4)b(x^{-1})+u\rangle$,
where
$b(x)\in (x^2+2x+4)^{2}(\mathcal{K}_2/\langle (x^2+2x+4)^{4}\rangle)$, i.e.,

\noindent
 $b(x)\in \{(g_0+g_1x)\cdot(x^2+2x+4)^2+(h_0+h_1x)\cdot (x^2+2x+4)^3\mid g_0,g_1,h_0,h_1\in \mathbb{F}_5\}$.

\vskip 2mm \par
  (ii-2) $5^{2\cdot 2}+5^{2\cdot 1}+5^{2\cdot 1}+5^{2\cdot 0}=676$ pairs:

\par
  $C_2=\langle (x^2+2x+4)^{k+1}b(x)+u(x^2+2x+4)^k\rangle$ and
$$C_{4}=\langle 3 x^{28}(x^2+3x+4)b(x^{-1})+u, (x^2+3x+4)^{5-k}\rangle,$$
where $b(x)\in (x^2+2x+4)^{\lceil \frac{5-k}{2}\rceil-1}(\mathcal{K}_2/\langle (x^2+2x+4)^{5-k-1}\rangle)$ and $1\leq k\leq 4$.

\vskip 2mm \par
  (ii-3) $6$ pairs: $C_2=\langle (x^2+2x+4)^k\rangle$ and $C_{4}=\langle (x^2+3x+4)^{5-k}\rangle$, where $0\leq k\leq 5$.

\vskip 2mm \par
  (ii-4) $5^{2\cdot 0}+5^{2\cdot 1}+5^{2\cdot 1}+5^{2\cdot 2}=676$ pairs:

\par
  $C_2=\langle (x^2+2x+4)b(x)+u,(x^2+2x+4)^{t}\rangle$ and
  $$C_{4}=\langle 3 x^{28}(x^2+3x+4)^{5-t+1}b(x^{-1})+u(x^2+3x+4)^{5-t}\rangle,$$
where
$b(x)\in (x^2+2x+4)^{\lceil \frac{t}{2}\rceil-1}(\mathcal{K}_2/\langle (x^2+2x+4)^{t-1}\rangle)$,
$1\leq t\leq 4$.

\vskip 2mm \par
  (ii-5) $(5^{2\cdot 0}+5^{2\cdot 1}+5^{2\cdot 1})+(5^{2\cdot 0}+5^{2\cdot 1})+5^{2\cdot 0}=78$ pairs:

  $C_2=\langle (x^2+2x+4)^{k+1}b(x)+u(x^2+2x+4)^k,(x^2+2x+4)^{k+t}\rangle$ and
  $$C_{4}=\langle 3 x^{28}(x^2+3x+4)^{5-k-t+1}b(x^{-1})+u(x^2+3x+4)^{5-k-t},(x^2+3x+4)^{5-k}\rangle,$$
  where
$b(x)\in (x^2+2x+4)^{\lceil \frac{t}{2}\rceil-1}(\mathcal{K}_2/\langle (x^2+2x+4)^{t-1}\rangle)$,
$1\leq t\leq 5-k-1$ and $1\leq k\leq 3$.

\vskip 3mm \par
  \begin{remark}  Using conclusions in \cite{s3}, one can determine all distinct negacyclic codes
over $\mathbb{F}_{p^m}+u\mathbb{F}_{p^m}$ of length $2p^sl^t$ and their dual codes, for any odd prime $l$ coprime to $p$ and positive integer
$t$.
\end{remark}


\section{$\lambda$-constacyclic codes over $R$ of length $np^s$ for $n=1,2,4$}

\noindent \par
In this section, we give the representations for $\lambda$-constacyclic codes
over $R=\mathbb{F}_{p^m}+u\mathbb{F}_{p^m}$ of length $np^s$ for $n=1,2,4$ and arbitrary
element $\lambda\in\mathbb{F}_{p^m}^\times$. Let $\lambda_0\in\mathbb{F}_{p^m}^\times$
satisfying $\lambda_0^{p^s}=\lambda$. Then $x^{np^s}-\lambda=(x^n-\lambda_0)^{p^s}$
and $x^{np^s}-\lambda^{-1}=(x^n-\lambda_0^{-1})^{p^s}$.

\vskip 3mm \noindent
  {\bf Case I: $n=1$.}
  In this case, $x-\lambda_0$ is irreducible in $\mathbb{F}_{p^m}[x]$. Denote $y=x-\lambda_0$. Then
$\mathcal{A}=\mathbb{F}_{p^m}[x]/\langle x^{p^s}-\lambda_0\rangle
=\mathbb{F}_{p^m}+y\mathbb{F}_{p^m}+\ldots+y^{p^s-1}\mathbb{F}_{p^m}
\ (y^{p^s}=0).$
Hence
$y^\alpha (\mathcal{A}/\langle y^\beta\rangle)
=\{\sum_{i=\alpha}^{\beta-1}a_iy^i\mid a_i\in \mathbb{F}_{p^m}, \ \alpha\leq i\leq \beta-1\}$ and $y^\alpha (\mathcal{A}/\langle y^\alpha\rangle)=\{0\}$ for all integers $0\leq \alpha<\beta\leq p^s-1$.
Denote $z=\widetilde{x-\lambda_0}=1-\lambda_0x$. Then by Corollary 3.10 and Theorem 4.4 we deduce the following conclusion.

\begin{corollary}\label{co7.1}
\textit{Denote $y=x-\lambda_0$, $z=1-\lambda_0x$ and $q=p^m$. Then all distinct $\lambda$-constacyclic codes over $\mathbb{F}_{p^m}+u\mathbb{F}_{p^m}$ of length $p^s$ and their dual codes are given by
the following table}:
\begin{center}
\begin{tabular}{lll}\hline
 N     &     $\mathcal{C}$ $({\rm mod} \ x^{p^s}-\lambda)$, $\mathcal{C}^{\bot}$ $({\rm mod} \ x^{p^s}-\lambda^{-1})$ & $|\mathcal{C}|$  \\ \hline
$q^{p^s-\lceil\frac{p^s}{2}\rceil}$ & $\bullet$ $\mathcal{C}=\langle y b(x)+u\rangle$
   & $q^{p^s}$  \\
   & $\mathcal{C}^{\bot}=\langle z\overline{b(x)}+u\rangle$   & \\
   & $b(x)\in y^{\lceil \frac{p^s}{2}\rceil-1}(\mathcal{A}/\langle y^{p^s-1}\rangle)$ &  \\
$\sum_{k=1}^{p^s-1}q^{p^s-k-\lceil\frac{1}{2}(p^s-k)\rceil}$
  & $\bullet$  $\mathcal{C}=\langle y^{k+1}b(x)+u y^k\rangle$ & $q^{p^s-k}$  \\
  & $\mathcal{C}^{\bot}=\langle z\overline{b(x)}+u,z^{p^s-k}\rangle$   & \\
  & $b(x)\in y^{\lceil \frac{p^s-k}{2}\rceil-1}(\mathcal{A}/\langle y^{p^s-k-1}\rangle),$ &  \\
  & $1\leq k\leq p^s-1$  & \\
$p^s+1$ & $\bullet$  $\mathcal{C}=\langle y^k\rangle$ & $q^{2(p^s-k)}$  \\
  & $\mathcal{C}^{\bot}=\langle z^{p^s-k}\rangle$   & \\
  & $0\leq k\leq p^s$ & \\
$\sum_{t=1}^{p^s-1}q^{t-\lceil\frac{t}{2}\rceil}$ & $\bullet$  $\mathcal{C}=\langle yb(x)+u,y^{t}\rangle$
  & $q^{2p^s-t}$  \\
  & $\mathcal{C}^{\bot}=\langle z^{p^s-t+1}\overline{b(x)}+uz^{p^s-t}\rangle$   & \\
  & $b(x)\in y^{\lceil\frac{t}{2}\rceil-1}(\mathcal{A}/\langle y^{t-1}\rangle)$, $1\leq t\leq p^s-1$ & \\
$\sum_{k=1}^{p^s-2}\sum_{t=1}^{p^s-k-1}q^{t-\lceil\frac{t}{2}\rceil}$
  & $\bullet$  $\mathcal{C}=\langle y^{k+1}b(x)+u y^k,y^{k+t}\rangle$
  & $q^{2p^s-2k-t}$  \\
  & $\mathcal{C}^{\bot}=\langle z^{p^s-k-t+1}\overline{b(x)}+uz^{p^s-k-t},z^{p^s-k}\rangle$   & \\
  & $b(x)\in y^{\lceil\frac{t}{2}\rceil-1}(\mathcal{A}/\langle y^{t-1}\rangle)$,
& \\
  & $1\leq t\leq p^s-k-1$, $1\leq k\leq p^s-2$ & \\
\hline
\end{tabular}
\end{center}

\noindent
in which $x^{-1}=\lambda x^{p^s-1}$ and $\overline{b(x)}=-\lambda x^{p^s-1}b(x^{-1})$ $({\rm mod} \
x^{p^s}-\lambda^{-1})$.

\par
  Moreover, the number of $\lambda$-constacyclic codes over $\mathbb{F}_{p^m}+u\mathbb{F}_{p^m}$ of length $p^s$ is equal to
$N_{(p^m,1,p^{s})}=\left\{\begin{array}{ll}\sum_{i=0}^{2^{s-1}}(1+4i)2^{(2^{s-1}-i)m}, & {\rm if} \ p=2; \cr & \cr
                                        \sum_{i=0}^{\frac{p^s-1}{2}}(3+4i)p^{(\frac{p^s-1}{2}-i)m}, & {\rm if} \ p \ {\rm is} \ {\rm odd}. \end{array}\right.$
\end{corollary}

\vskip 3mm \noindent
  {\bf Case II: $n=2$ (and $p$ is odd).}
   We have the following three subcases.

\vskip 3mm\par
   {\bf Case (II-1)}: $\lambda_0$ is not a square, i.e.
  $\lambda_0\in \mathbb{F}_{p^m}^\times\setminus (\mathbb{F}_{p^m}^\times)^2$.

\par
  In this case, $x^2-\lambda_0$ is irreducible in $\mathbb{F}_{p^m}[x]$ and
$\mathcal{A}=\mathbb{F}_{p^m}[x]/\langle (x^2-\lambda_0)^{p^s}\rangle$. We
denote $\mathcal{T}^{(2)}=\{c_0+c_1x\mid c_0,c_1\in \mathbb{F}_{p^m}\}$,
$$(x^2-\lambda_0)^\alpha (\mathcal{A}/\langle (x^2-\lambda_0)^\beta\rangle)
=\{\sum_{i=\alpha}^{\beta-1}\xi_i(x^2-\lambda_0)^i\mid \xi_\alpha,\ldots,\xi_{\beta-1}\in \mathcal{T}^{(2)}\}$$
and $(x^2-\lambda_0)^\alpha (\mathcal{A}/\langle (x^2-\lambda_0)^\alpha\rangle)=\{0\}$
for all integers $0\leq \alpha<\beta\leq p^s-1$.
Then by Corollary 3.10 and Theorem 4.4 we deduce the following conclusion.

\begin{corollary}\label{co7.2}
\textit{Let $p$ be odd and $\lambda_0\in \mathbb{F}_{p^m}^\times\setminus (\mathbb{F}_{p^m}^\times)^2$. Denote $y=x^2-\lambda_0$, $z=\widetilde{x^2-\lambda_0}=1-\lambda_0x^2$ and $q=p^m$. Then all distinct $\lambda$-constacyclic codes over $\mathbb{F}_{p^m}+u\mathbb{F}_{p^m}$ of length $2p^s$ and their dual codes are given by
the following table}:
{\small
\begin{center}
\begin{tabular}{lll}\hline
 N     &     $\mathcal{C}$ $({\rm mod} \ x^{2p^s}-\lambda)$, $\mathcal{C}^{\bot}$ $({\rm mod} \ x^{2p^s}-\lambda^{-1})$ & $|\mathcal{C}|$  \\ \hline
$q^{2(p^s-\lceil\frac{p^s}{2}\rceil)}$ & $\bullet$ $\mathcal{C}=\langle y b(x)+u\rangle$
   & $q^{2p^s}$  \\
   & $\mathcal{C}^{\bot}=\langle z\overline{b(x)}+u\rangle$   & \\
   & $b(x)\in y^{\lceil \frac{p^s}{2}\rceil-1}(\mathcal{A}/\langle y^{p^s-1}\rangle)$ &  \\
$\sum_{k=1}^{p^s-1}q^{2(p^s-k-\lceil\frac{1}{2}(p^s-k)\rceil)}$
  & $\bullet$  $\mathcal{C}=\langle y^{k+1}b(x)+u y^k\rangle$ & $q^{2(p^s-k)}$  \\
  & $\mathcal{C}^{\bot}=\langle z\overline{b(x)}+u,z^{p^s-k}\rangle$   & \\
  & $b(x)\in y^{\lceil \frac{p^s-k}{2}\rceil-1}(\mathcal{A}/\langle y^{p^s-k-1}\rangle),$ &  \\
  & $1\leq k\leq p^s-1$  & \\
$p^s+1$ & $\bullet$  $\mathcal{C}=\langle y^k\rangle$ & $q^{4(p^s-k)}$  \\
  & $\mathcal{C}^{\bot}=\langle z^{p^s-k}\rangle$   & \\
  & $0\leq k\leq p^s$ & \\
$\sum_{t=1}^{p^s-1}q^{2(t-\lceil\frac{t}{2}\rceil)}$ & $\bullet$  $\mathcal{C}=\langle yb(x)+u,y^{t}\rangle$
  & $q^{2(2p^s-t)}$  \\
  & $\mathcal{C}^{\bot}=\langle z^{p^s-t+1}\overline{b(x)}+uz^{p^s-t}\rangle$   & \\
  & $b(x)\in y^{\lceil\frac{t}{2}\rceil-1}(\mathcal{A}/\langle y^{t-1}\rangle)$, $1\leq t\leq p^s-1$ & \\
$\sum_{k=1}^{p^s-2}\sum_{t=1}^{p^s-k-1}q^{2(t-\lceil\frac{t}{2}\rceil)}$
  & $\bullet$  $\mathcal{C}=\langle y^{k+1}b(x)+u y^k,y^{k+t}\rangle$
  & $q^{2(2p^s-2k-t)}$  \\
  & $\mathcal{C}^{\bot}=\langle z^{p^s-k-t+1}\overline{b(x)}+uz^{p^s-k-t},z^{p^s-k}\rangle$   & \\
  & $b(x)\in y^{\lceil\frac{t}{2}\rceil-1}(\mathcal{A}/\langle y^{t-1}\rangle)$,
& \\
  & $1\leq t\leq p^s-k-1$, $1\leq k\leq p^s-2$ & \\
\hline
\end{tabular}
\end{center} }

\noindent
in which $x^{-1}=\lambda x^{2p^s-1}$ and $\overline{b(x)}=-\lambda x^{2p^s-2}b(x^{-1})$ $({\rm mod} \
x^{2p^s}-\lambda^{-1})$.

\par
  Moreover, the number of $\lambda$-constacyclic codes over $\mathbb{F}_{p^m}+u\mathbb{F}_{p^m}$ of length $2p^s$ is equal to
$N_{(p^m,2,p^{s})}=\sum_{i=0}^{\frac{p^s-1}{2}}(3+4i)p^{2(\frac{p^s-1}{2}-i)m}$.
\end{corollary}

\vskip 3mm\par
   {\bf Case (II-2)}: $\lambda_0$ is a square, i.e. $\lambda_0=\gamma^2$
for some $\gamma\in \mathbb{F}_{p^m}^\times$.

\par
  We have that $x^2-\lambda_0=(x-\gamma)(x+\gamma)$,
$x^{2s}-\lambda=(x^{p^s}-\gamma^{p^s})(x^{p^s}+\gamma^{p^s})$   and
$\frac{1}{2}\gamma^{-p^s}(x^{p^s}+\gamma^{p^s})-\frac{1}{2}\gamma^{-p^s}(x^{p^s}-\gamma^{p^s})=1$. Denote
$$\varepsilon_1(x)=\frac{1}{2}\gamma^{-p^s}(x^{p^s}+\gamma^{p^s}) \
{\rm and} \ \varepsilon_2(x)=-\frac{1}{2}\gamma^{-p^s}(x^{p^s}-\gamma^{p^s}).$$
By Corollary \ref{co3.9}, all distinct $\lambda$-constacyclic codes over
$\mathbb{F}_{p^m}+u\mathbb{F}_{p^m}$ of length $2p^s$ are given by
$$\mathcal{C}=\varepsilon_1(x)C_1\oplus \varepsilon_2(x)C_2,$$
where $C_1$ is an ideal of the ring
$(\mathbb{F}_{p^m}+u\mathbb{F}_{p^m})[x]/\langle x^{p^s}-\gamma^{p^s}\rangle$
and $C_2$ is an ideal of the ring $(\mathbb{F}_{p^m}+u\mathbb{F}_{p^m})[x]/\langle x^{p^s}-(-\gamma)^{p^s}\rangle$. In particular, $C_1$ is a $\gamma^{p^s}$-constacyclic code
over $\mathbb{F}_{p^m}+u\mathbb{F}_{p^m}$ of length $p^s$ and $C_2$ is a $(-\gamma)^{p^s}$-constacyclic code
over $\mathbb{F}_{p^m}+u\mathbb{F}_{p^m}$ of length $p^s$. Hence

\par
  $\diamondsuit$ $C_1$ is listed by Corollary \ref{co7.1} substituted $\gamma$ for $\lambda_0$;

\par
  $\diamondsuit$ $C_2$ is listed by Corollary \ref{co7.1} substituted $-\gamma$ for $\lambda_0$.

\noindent
  Therefore, the number of $\lambda$-constacyclic codes over
$\mathbb{F}_{p^m}+u\mathbb{F}_{p^m}$ of length $2p^s$ is equal to $N^2_{(p^m,1,p^s)}=\left(\sum_{i=0}^{\frac{p^s-1}{2}}(3+4i)p^{(\frac{p^s-1}{2}-i)m}\right)^2$.

\vskip 3mm \noindent
  {\bf Case III: $n=4$ (and $p$ is odd).}
We have one of the following four subcases.

\vskip 3mm\par
   {\bf Case (III-1)}: $p^m\equiv 1$ (mod 4) and $\lambda_0$ is not a square, i.e.
  $\lambda_0\in \mathbb{F}_{p^m}^\times\setminus (\mathbb{F}_{p^m}^\times)^2$.

\par
  In this case, $x^4-\lambda_0$ is irreducible in $\mathbb{F}_{p^m}[x]$ and
$\mathcal{A}=\mathbb{F}_{p^m}[x]/\langle (x^4-\lambda_0)^{p^s}\rangle$. We
denote $\mathcal{T}^{(4)}=\{c_0+c_1x+c_2x^2+c_3x^3\mid c_0,c_1,c_2,c_3\in \mathbb{F}_{p^m}\}$,
$$(x^4-\lambda_0)^\alpha (\mathcal{A}/\langle (x^4-\lambda_0)^\beta\rangle)
=\{\sum_{i=\alpha}^{\beta-1}\xi_i(x^4-\lambda_0)^i\mid \xi_\alpha,\ldots,\xi_{\beta-1}\in \mathcal{T}^{(4)}\}$$
and $(x^4-\lambda_0)^\alpha (\mathcal{A}/\langle (x^4-\lambda_0)^\alpha\rangle)=\{0\}$
for all integers $0\leq \alpha<\beta\leq p^s-1$.
Then by Corollary 3.10 and Theorem 4.4 we deduce the following conclusion.

\begin{corollary}\label{co7.3}
\textit{Let $p$ be odd, $p^m\equiv 1$ $({\rm mod} \ 4)$ and $\lambda_0\in
\mathbb{F}_{p^m}^\times\setminus(\mathbb{F}_{p^m}^\times)^2$. Denote $y=x^4-\lambda_0$, $z=\widetilde{x^4-\lambda_0}=1-\lambda_0x^4$ and $q=p^m$. Then all distinct $\lambda$-constacyclic codes over $\mathbb{F}_{p^m}+u\mathbb{F}_{p^m}$ of length $4p^s$ and their dual codes are given by
the following table}:
{\small
\begin{center}
\begin{tabular}{lll}\hline
 N     &     $\mathcal{C}$ $({\rm mod} \ x^{p^s}-\lambda)$, $\mathcal{C}^{\bot}$ $({\rm mod} \ x^{p^s}-\lambda^{-1})$ & $|\mathcal{C}|$  \\ \hline
$q^{4(p^s-\lceil\frac{p^s}{2}\rceil)}$ & $\bullet$ $\mathcal{C}=\langle y b(x)+u\rangle$
   & $q^{4p^s}$  \\
   & $\mathcal{C}^{\bot}=\langle z\overline{b(x)}+u\rangle$   & \\
   & $b(x)\in y^{\lceil \frac{p^s}{2}\rceil-1}(\mathcal{A}/\langle y^{p^s-1}\rangle)$ &  \\
$\sum_{k=1}^{p^s-1}q^{4(p^s-k-\lceil\frac{1}{2}(p^s-k)\rceil)}$
  & $\bullet$  $\mathcal{C}=\langle y^{k+1}b(x)+u y^k\rangle$ & $q^{4(p^s-k)}$  \\
  & $\mathcal{C}^{\bot}=\langle z\overline{b(x)}+u,z^{p^s-k}\rangle$   & \\
  & $b(x)\in y^{\lceil \frac{p^s-k}{2}\rceil-1}(\mathcal{A}/\langle y^{p^s-k-1}\rangle),$ &  \\
  & $1\leq k\leq p^s-1$  & \\
$p^s+1$ & $\bullet$  $\mathcal{C}=\langle y^k\rangle$ & $q^{8(p^s-k)}$  \\
  & $\mathcal{C}^{\bot}=\langle z^{p^s-k}\rangle$   & \\
  & $0\leq k\leq p^s$ & \\
$\sum_{t=1}^{p^s-1}q^{4(t-\lceil\frac{t}{2}\rceil)}$ & $\bullet$  $\mathcal{C}=\langle yb(x)+u,y^{t}\rangle$
  & $q^{4(2p^s-t)}$  \\
  & $\mathcal{C}^{\bot}=\langle z^{p^s-t+1}\overline{b(x)}+uz^{p^s-t}\rangle$   & \\
  & $b(x)\in y^{\lceil\frac{t}{2}\rceil-1}(\mathcal{A}/\langle y^{t-1}\rangle)$, $1\leq t\leq p^s-1$ & \\
$\sum_{k=1}^{p^s-2}\sum_{t=1}^{p^s-k-1}q^{4(t-\lceil\frac{t}{2}\rceil)}$
  & $\bullet$  $\mathcal{C}=\langle y^{k+1}b(x)+u y^k,y^{k+t}\rangle$
  & $q^{4(2p^s-2k-t)}$  \\
  & $\mathcal{C}^{\bot}=\langle z^{p^s-k-t+1}\overline{b(x)}+uz^{p^s-k-t},z^{p^s-k}\rangle$   & \\
  & $b(x)\in y^{\lceil\frac{t}{2}\rceil-1}(\mathcal{A}/\langle y^{t-1}\rangle)$,
& \\
  & $1\leq t\leq p^s-k-1$, $1\leq k\leq p^s-2$ & \\
\hline
\end{tabular}
\end{center} }

\noindent
in which $x^{-1}=\lambda x^{4p^s-1}$ and $\overline{b(x)}=-\lambda x^{4p^s-4}b(x^{-1})$ $({\rm mod} \
x^{4p^s}-\lambda^{-1})$.

\par
  Moreover, the number of $\lambda$-constacyclic codes over $\mathbb{F}_{p^m}+u\mathbb{F}_{p^m}$ of length $4p^s$ is equal to
$N_{(p^m,4,p^{s})}=\sum_{i=0}^{\frac{p^s-1}{2}}(3+4i)p^{4(\frac{p^s-1}{2}-i)m}$.
\end{corollary}

\par
  For example, we know that $\mathbb{F}_{13}^\times\setminus(\mathbb{F}_{13}^\times)^2=\{2,5,6,7,8,11\}$.
For any $\lambda\in \mathbb{F}_{13}^\times\setminus(\mathbb{F}_{13}^\times)^2$,
$x^4-\lambda$ is irreducible in $\mathbb{F}_{13}[x]$ and the number
$\lambda$-constacyclic codes over $\mathbb{F}_{13}+u\mathbb{F}_{13}$ of length $4\cdot 13=52$ is equal to
$N_{(13,4,13)}=\sum_{i=0}^{6}(3+4i)13^{4(6-i)}=1628535353189467891702213785$.
Roughly, we have $16\cdot 10^{26}<N_{(13,4,13)}<17\cdot 10^{26}$.

\vskip 3mm\par
   {\bf Case (III-2)}: $p^m\equiv 1$ (mod $4$) and $\lambda_0$ is a square, i.e.
  $\lambda_0\in (\mathbb{F}_{p^m}^\times)^2$.

\par
  In this case, let $\zeta$ be a primitive element
of $\mathbb{F}_{p^m}$. Then ${\rm ord}(\zeta)=p^m-1$,
 and $(\mathbb{F}_{p^m}^\times)^2=(\mathbb{F}_{p^m}^\times)^4
\cup \zeta^2(\mathbb{F}_{p^m}^\times)^4$. By $p^m\equiv 1$ (mod $4$),
we know that $4$ is a factor of $p^m-1$. Denote $\xi=\zeta^{\frac{p^m-1}{4}}$. Then
$\xi$ is a $4$th primitive root of unity and $\xi^2=-1$. Now, we have one of the following two cases:

\par
  {\bf Case (III-2-1)}: $\lambda_0\in (\mathbb{F}_{p^m}^\times)^4$.

\par
   Let $\lambda_0=w^4$ where $w\in \mathbb{F}_{p^m}^\times$. Then $x^4-\lambda_0
=\prod_{j=0}^3(x-w\xi^j)$. This implies that $x^{4p^s}-\lambda
=\prod_{j=0}^3(x^{p^s}-(w\xi^j)^{p^s})$. Set
$$\theta_j(x)=\prod_{0\leq l\leq 3, \ l\neq j}\frac{x^{p^s}-(w\xi^l)^{p^s}}{(w\xi^j)^{p^s}-(w\xi^l)^{p^s}}, \ j=0,1,2,3.$$
By Corollary \ref{co3.9}, all distinct $\lambda$-constacyclic codes over
$\mathbb{F}_{p^m}+u\mathbb{F}_{p^m}$ of length $4p^s$ are given by
$$\mathcal{D}=\theta_1(x)D_1\oplus \theta_2(x)D_2\oplus \theta_3(x)D_3\oplus \theta_4(x)D_4,$$
where $D_j$ is an ideal of the ring $(\mathbb{F}_{p^m}+u\mathbb{F}_{p^m})[x]/\langle x^{p^s}-(w\xi^j)^{p^s}\rangle$, i.e. $D_j$ is a $(w\xi^j)^{p^s}$-constacyclic code
over $\mathbb{F}_{p^m}+u\mathbb{F}_{p^m}$ of length $p^s$ for all $j$. Therefore,

\par
  $\diamondsuit$ $D_j$ is given by Corollary 7.1 substituted $w\xi^j$ and $(w\xi^j)^{p^s}$ for $\lambda_0$
and $\lambda$ respectively, $j=0,1,2,3$.

\par
   Hence the number of $\lambda$-constacyclic codes over $\mathbb{F}_{p^m}+u\mathbb{F}_{p^m}$ of length $4p^s$ is equal to
$N_{(p^m,1,p^s)}^4=\left(\sum_{i=0}^{\frac{p^s-1}{2}}(3+4i)p^{(\frac{p^s-1}{2}-i)m}\right)^4.$

\par
  For example, we know that $(\mathbb{F}_{13}^\times)^2=\{1,3,4,9,10,12\}$,
$(\mathbb{F}_{13}^\times)^4=\{1,3,9\}$,
  $x^4-1=(x+1)(x+5)(x+12)(x+8)$, $x^4-3=(x+10)(x+3)(x+2)(x+11)$ and $x^4-9=(x+7)(x+6)(x+9)(x+4)$.
Hence for any $\lambda\in \{1,3,9\}$, the number
$\lambda$-constacyclic codes over $\mathbb{F}_{13}+u\mathbb{F}_{13}$ of length $4\cdot 13=52$ is equal to
$N_{(13,1,13)}^4=(\sum_{i=0}^{6}(3+4i)13^{6-i})^4=92300403860395414742363374161$. Roughly, we have that $923\cdot 10^{26}<N_{(13,1,13)}^4<924\cdot 10^{26}$.

\par
  {\bf Case (III-2-2)}: $\lambda_0\in \zeta^2(\mathbb{F}_{p^m}^\times)^4$.

\par
   Let $\lambda_0=\zeta^2w^4$ where $w\in \mathbb{F}_{p^m}^\times$. Then $x^4-\lambda_0
=(x^2-w^2\zeta)(x^2+w^2\zeta)$. This implies $x^{4p^s}-\lambda
=(x^{2p^s}-(w^2\zeta)^{p^s})(x^{2p^s}-(-w^2\zeta)^{p^s})$. By $-1=\xi^2$, we see that both $w^2\zeta$ and
$-w^2\zeta$ are not squares in $\mathbb{F}_{p^m}^\times$. This implies
that $x^2-w^2\zeta$ and $x^2+w^2\zeta$ are irreducible in $\mathbb{F}_{p^m}[x]$.
Obviously $\frac{1}{2}(w^2\zeta)^{-p^s}(x^{2p^s}+(w^2\zeta)^{p^s})
-\frac{1}{2}(w^2\zeta)^{-p^s}(x^{2p^s}-(w^2\zeta)^{p^s})=1$. Now, set
$$\theta_1(x)=\frac{1}{2}(w^2\zeta)^{-p^s}(x^{2p^s}+(w^2\zeta)^{p^s})
\ {\rm and} \ \theta_2(x)=-\frac{1}{2}(-w^2\zeta)^{-p^s}(x^{2p^s}-(w^2\zeta)^{p^s}).$$
By Corollary \ref{co3.9}, all distinct $\lambda$-constacyclic codes over
$\mathbb{F}_{p^m}+u\mathbb{F}_{p^m}$ of length $4p^s$ are given by
$$\mathcal{D}=\theta_1(x)D_1\oplus \theta_2(x)D_2,$$
where $D_1$ is an ideal of the ring $(\mathbb{F}_{p^m}+u\mathbb{F}_{p^m})[x]/\langle x^{2p^s}-(w^2\zeta)^{p^s}\rangle$ and $D_2$ is an ideal of the ring $(\mathbb{F}_{p^m}+u\mathbb{F}_{p^m})[x]/\langle x^{2p^s}+(w^2\zeta)^{p^s}\rangle$.
Therefore,

\par
  $\diamondsuit$ $D_1$ is a $(w^2\zeta)^{p^s}$-constacyclic code over
$\mathbb{F}_{p^m}+u\mathbb{F}_{p^m}$ of length $2p^s$ which can be given by Corollary 7.2 substituted $w^2\zeta$ and $(w^2\zeta)^{p^s}$ for $\lambda_0$ and $\lambda$ respectively.

\par
  $\diamondsuit$ $D_2$ is a $-(w^2\zeta)^{p^s}$-constacyclic code over
$\mathbb{F}_{p^m}+u\mathbb{F}_{p^m}$ of length $2p^s$ which can be given by Corollary 7.2 substituted $-w^2\zeta$ and $-(w^2\zeta)^{p^s}$ for $\lambda_0$ and $\lambda$ respectively.

\par
   Hence the number of $\lambda$-constacyclic codes over $\mathbb{F}_{p^m}+u\mathbb{F}_{p^m}$ of length $4p^s$ is equal to
$N_{(p^m,2,p^s)}^2=\left(\sum_{i=0}^{\frac{p^s-1}{2}}(3+4i)p^{2(\frac{p^s-1}{2}-i)m}\right)^2.$

\par
  For example, we have
$(\mathbb{F}_{13}^\times)^2\setminus (\mathbb{F}_{13}^\times)^4=\{4,10,12\}$,
$x^4-4=(x^2+11)(x^2+2)$, $x^4-10=(x^2+6)(x^2+7)$ and $x^4-12=(x^2+8)(x^2+5)$.
Hence for any $\lambda\in \{4,10,12\}$, the number
$\lambda$-constacyclic codes over $\mathbb{F}_{13}+u\mathbb{F}_{13}$ of length $4\cdot 13=52$ is equal to
$N_{(13,2,13)}^2=(\sum_{i=0}^{6}(3+4i)13^{2(6-i)})^2=5022317475223730190748850625$. Roughly, we have $50\cdot 10^{26}<N_{(13,2,13)}^2<51\cdot 10^{26}$.

\vskip 3mm\par
   {\bf Case (III-3)}: $p^m\equiv 3$ (mod 4) and $\lambda_0$ is not a square, i.e.
  $\lambda_0\in \mathbb{F}_{p^m}^\times\setminus (\mathbb{F}_{p^m}^\times)^2$.

\par
  In this case, $x^4-\lambda_0$ is reducible in $\mathbb{F}_{p^m}[x]$ (cf. Wan [18] Theorem 10.7).
It is evident that $x^4-\lambda_0$ has no divisors of degree $1$ and $3$.  Suppose that $x^2-c$ is a factor
of $x^4-\lambda_0$ for some $c\in \mathbb{F}_{p^m}$. Then we have $\lambda_0=c^2$, which
contradicts that $\lambda_0$ is not a square. Hence there are elements $a,b,c,d\in \mathbb{F}_{p^m}^\times$
such that
$$x^4-\lambda_0=f_1(x)f_2(x) \ {\rm where} \
f_1(x)=x^2+ax+b \ {\rm and} \ f_2(x)=x^2+bx+d,$$
and that $f_1(x)$ and $f_2(x)$ are irreducible in $\mathbb{F}_{p^m}[x]$
satisfying ${\rm gcd}(f_1(x),f_2(x))=1$. We find polynomials $v(x),w(x)\in \mathbb{F}_{p^m}[x]$
satisfying $v(x)f_2(x)+w(x)f_1(x)=1$. This implies $v(x)^{p^s}f_2(x)^{p^s}+w(x)^{p^s}f_1(x)^{p^s}=1$.  Then set
$$\theta_1(x)=v(x)^{p^s}f_2(x)^{p^s} \ {\rm and} \ \theta_2(x)=w(x)^{p^s}f_1(x)^{p^s}.$$
By Corollary \ref{co3.9}, all distinct $\lambda$-constacyclic codes over
$\mathbb{F}_{p^m}+u\mathbb{F}_{p^m}$ of length $4p^s$ are given by
$$\mathcal{D}=\theta_1(x)D_1\oplus \theta_2(x)D_2,$$
where $D_j$ is an ideal of the ring $(\mathbb{F}_{p^m}[x]/\langle f_j(x)^{p^s}\rangle)+u(\mathbb{F}_{p^m}[x]/\langle f_j(x)^{p^s}\rangle)=(\mathbb{F}_{p^m}+u\mathbb{F}_{p^m})[x]/\langle f_j(x)^{p^s}\rangle$
for $j=1,2$. From this and by Theorem \ref{th3.8}, we deduce that
\par
  $\diamondsuit$ $C_j$ is given by Corollary 7.2 substituted $f_j(x)$, $f_j(x)^{p^s}$
and $\widetilde{f}_j(x)^{p^s}$ for $x^2-\lambda_0$, $x^{2p^s}-\lambda$
and $x^{2p^s}-\lambda^{-1}$ respectively.

\par
   Therefore, the number of $\lambda$-constacyclic codes over $\mathbb{F}_{p^m}+u\mathbb{F}_{p^m}$ of length $4p^s$ is equal to
$N_{(p^m,2,p^s)}^2=\left(\sum_{i=0}^{\frac{p^s-1}{2}}(3+4i)p^{2(\frac{p^s-1}{2}-i)m}\right)^2.$

\par
  For example, we have $\mathbb{F}_{19}^\times\setminus (\mathbb{F}_{19}^\times)^2=\{2,3,8,10,12,13,14,15,18\}$ and that

\par
  $x^4-2=(x^2+8x+13)(x^2+11x+13)$,
  $x^4-3=(x^2+10x+12)(x^2+9x+12)$,

\par
  $x^4-8=(x^2+12x+15)(x^2+7x+15)$,
  $x^4-10=(x^2+5x+3)(x^2+14x+3)$,

\par
  $x^4-12=(x^2+4x+8)(x^2+15x+8)$,
  $x^4-13=(x^2+16x+14)(x^2+3x+14)$,

\par
  $x^4-14=(x^2+18x+10)(x^2+x+10)$,
  $x^4-15=(x^2+17x+2)(x^2+2x+2)$,

\par
  $x^4-18=(x^2+13x+18)(x^2+6x+18)$.

\noindent
  Therefore, for any $\lambda\in \mathbb{F}_{19}^\times\setminus(\mathbb{F}_{19}^\times)^2$, the number
of $\lambda$-constacyclic codes over $\mathbb{F}_{19}+u\mathbb{F}_{19}$ of
length $4\cdot 19=76$ is equal to
$N_{(19,2,19)}^2=(\sum_{i=0}^{9}(3+4i)19^{2(9-i)})^2
=98853624946129979125010756140470464728908752100$. Roughly, we have
$98\cdot 10^{45}<N_{(19,2,19)}^2<99\cdot 10^{45}$.

\par
   {\bf Case (III-4)}: $p^m\equiv 3$ (mod 4) and $\lambda_0$ is a square, i.e.
  $\lambda_0\in (\mathbb{F}_{p^m}^\times)^2$.

\par
  In this case, there exists $\eta\in \mathbb{F}_{p^m}^\times$ such that
$\lambda_0=\eta^2$. This implies $\lambda=\eta^{2p^s}$, $x^4-\lambda_0=(x^2-\eta)(x^2+\eta)$ and $x^{4p^s}-\lambda=(x^{2p^s}-\eta^{p^s})(x^{2p^s}+\eta^{p^s})$. Obviously,
we have $\frac{1}{2}\eta^{-p^s}(x^{2p^s}+\gamma^{p^s})-\frac{1}{2}\eta^{-p^s}(x^{2p^s}-\gamma^{p^s})=1$. Now, set
$$\theta_1(x)=\frac{1}{2}\eta^{-p^s}(x^{2p^s}+\gamma^{p^s}) \
{\rm and} \ \theta_2(x)=-\frac{1}{2}\eta^{-p^s}(x^{2p^s}-\gamma^{p^s}).$$
Then $\theta_1(x)^2=\theta_1(x)$, $\theta_2(x)^2=\theta_2(x)$, $\theta_1(x)\theta_2(x)=0$
and $\theta_1(x)+\theta_2(x)=1$ in the ring $\mathcal{A}=\mathbb{F}_{p^m}[x]/\langle x^{4p^s}-\lambda\rangle$.
Hence, all $\lambda$-constacyclic codes over $\mathbb{F}_{p^m}+u\mathbb{F}_{p^m}$ of length $4p^s$, i.e.
all ideals of the ring $(\mathbb{F}_{p^m}+u\mathbb{F}_{p^m})[x]/\langle x^{4p^s}-\eta^{2p^s}\rangle$,
are given by:
$$\mathcal{D}=\theta_1(x)D_1\oplus\theta_2(x)D_2,$$
where
   $D_1$ is an ideal of the ring $(\mathbb{F}_{p^m}+u\mathbb{F}_{p^m})[x]/\langle x^{2p^s}-\eta^{p^s}\rangle$, i.e. an $\eta^{p^s}$-constacyclic code over $\mathbb{F}_{p^m}+u\mathbb{F}_{p^m}$ of length $2p^s$,
and
   $D_2$ is an ideal of $(\mathbb{F}_{p^m}+u\mathbb{F}_{p^m})[x]/\langle x^{2p^s}+\eta^{p^s}\rangle$, i.e. an $(-\eta^{p^s})$-constacyclic code over $\mathbb{F}_{p^m}+u\mathbb{F}_{p^m}$ of length $2p^s$.

\par
  As $p^m\equiv 3$ (mod $4$), $-1$ is not a square of $\mathbb{F}_{p^m}^\times$. Hence we have one of the
following cases:

\par
   {\bf Case (III-4-1)}: $\eta$ is a square of $\mathbb{F}_{p^m}^\times$, i.e. $\eta=\delta^2$
for some $\delta\in \mathbb{F}_{p^m}^\times$.

\par
   In this case, $x^2-\eta=(x-\delta)(x+\delta)$ and $x^2+\eta$ is irreducible
in $\mathbb{F}_{p^m}[x]$ since $-\eta=(-1)\delta^2$ is not a square. Therefore,

\par
  $\triangleright$ $D_1$ is a $\eta^{p^s}$-constacyclic code over $\mathbb{F}_{p^m}+u\mathbb{F}_{p^m}$ of length $2p^s$ given by Case (II-2) substituted $\eta$ for $\lambda_0$;

\par
  $\triangleright$ $D_2$ is an $(-\eta^{p^s})$-constacyclic code over $\mathbb{F}_{p^m}+u\mathbb{F}_{p^m}$ of length $2p^s$ given by Case (II-1) substituted $-\eta$ for $\lambda_0$.

\par
   {\bf Case (III-4-2)}: $\eta$ is not a square of $\mathbb{F}_{p^m}^\times$.

\par
   In this case, $-\eta$ is a a square of $\mathbb{F}_{p^m}^\times$. Hence

\par
  $\triangleright$ $D_1$ is an $\eta^{p^s}$-constacyclic code over $\mathbb{F}_{p^m}+u\mathbb{F}_{p^m}$ of length $2p^s$ given by Case (II-1) substituted $\eta$ for $\lambda_0$.

\par
  $\triangleright$ $D_2$ is a $(-\eta^{p^s})$-constacyclic code over $\mathbb{F}_{p^m}+u\mathbb{F}_{p^m}$ of length $2p^s$ given by Case (II-2) substituted $-\eta$ for $\lambda_0$;

\par
  Therefore, the number of $\lambda$-constacyclic codes over $\mathbb{F}_{p^m}+u\mathbb{F}_{p^m}$ of length $4p^s$ is equal to
$$N_{(p^m,1,p^s)}^2N_{(p^m,2,p^s)}
=\left(\sum_{i=0}^{\frac{p^s-1}{2}}(3+4i)p^{(\frac{p^s-1}{2}-i)m}\right)^2
\left(\sum_{i=0}^{\frac{p^s-1}{2}}(3+4i)p^{2(\frac{p^s-1}{2}-i)m}\right).$$

\par
  For example, we have $(\mathbb{F}_{19}^\times)^2=\{1,4,5,6,7,9,11,16,17\}$ and that

\par
 $x^4-1=(x+1)(x+18)(x^2+1)$,
 $x^4-4=(x+13)(x+6)(x^2+17)$,

\par
 $x^4-5=(x+16)(x+3)(x^2+9)$,
 $x^4-6=(x+9)(x+10)(x^2+5)$,

\par
 $x^4-7=(x+12)(x+7)(x^2+11)$,
 $x^4-9=(x+15)(x+4)(x^2+16)$,

\par
 $x^4-11=(x+11)(x+8)(x^2+7)$,
 $x^4-16=(x+2)(x+17)(x^2+4)$,

\par
 $x^4-17=(x+14)(x+5)(x^2+6)$.

\noindent
  Therefore, for any $\lambda\in (\mathbb{F}_{19}^\times)^2$, the number
of $\lambda$-constacyclic codes over $\mathbb{F}_{19}+u\mathbb{F}_{19}$ of
length $4\cdot 19=76$ is equal to
\begin{eqnarray*}
N_{(19,1,19)}^2N_{(19,2,19)}
 &=&\left(\sum_{i=0}^{9}(3+4i)19^{9-i)}\right)^2
\left(\sum_{i=0}^{9}(3+4i)19^{2(9-i)}\right)\\
 &=&378733991979096789784301581334490215632932864000.
\end{eqnarray*}
Roughly, we have $378\cdot 10^{45}<N_{(19,1,19)}^2N_{(19,2,19)}<379\cdot 10^{45}$.

\vskip 3mm \par
  \begin{remark} (i) The dual code of each code in Cases (II-2), (III-2-1), (III-2-2), (III-3)
(III-4-1) and (III-4-2) can be determined by Theorem \ref{th4.4}. In particular,
self-dual cyclic codes and negacyclic codes over $\mathbb{F}_{p^m}+u\mathbb{F}_{p^m}$ of
length $np^s$ are given by Theorem \ref{th5.3} for arbitrary positive integer $n$ satisfying
${\rm gcd}(p,n)=1$.

\par
  (ii) Corollary \ref{co7.1} summarizes
the known results for cyclic codes and $\alpha$-constacyclic codes
over $\mathbb{F}_{2^m}+u\mathbb{F}_{2^m}$ with length $2^s$, where $\alpha\in \mathbb{F}_{2^m}^\times$, given by Theorem 4.4, Proposition 4.5,
Theorem 4.7, Theorem 4.9 and Theorems 5.3--5.5 in \cite{s8},
and the known results for
cyclic codes over $\mathbb{F}_{p^m}+u\mathbb{F}_{p^m}$ with length $p^s$ given by Theorems 5.4 and 5.7 in \cite{s9} for any prime $p$.

\par
  (iii) Corollary \ref{co7.2} summarizes
the known results in \cite{s10} Theorem 4.2, Proposition 4.3, Theorem 4.4, Lemma 4.6 and Theorems 4.10--4.12
for negacyclic codes over $\mathbb{F}_{p^m}+u\mathbb{F}_{p^m}$ of length $2p^s$ where $p^m\equiv 3$ (mod 4), and the known results in \cite{s7} Theorem 3.9, Proposition 3.10, Theorem 3.11, Lemma 3.13,
Theorem 3.14, Lemma 3.17, Theorems 3.18 and 3.19 for $\gamma$-constacyclic codes over $\mathbb{F}_{p^m}+u\mathbb{F}_{p^m}$ of length $2p^s$ where
$\gamma\in \mathbb{F}_{p^m}^\times\setminus(\mathbb{F}_{p^m}^\times)^2$.
Case (II-2) summarizes
the known results in \cite{s10} Theorem 3.1 for negacyclic codes over $\mathbb{F}_{p^m}+u\mathbb{F}_{p^m}$ of length $2p^s$ where $p^m\equiv 1$ (mod 4).

\par
  (iv) Corollary \ref{co7.3} summarizes
the known results in \cite{s11} Theorem 5.3, Proposition 5.4, Lemma 5.6,
Theorem 5.7 and Theorems 5.10--5.12 for $\gamma$-constacyclic codes over $\mathbb{F}_{p^m}+u\mathbb{F}_{p^m}$ of length $4p^s$ where
$\gamma\in \mathbb{F}_{p^m}^\times\setminus(\mathbb{F}_{p^m}^\times)^2$
and $p^m\equiv 1$ (mod $4$).

\par
  (v) The known results in \cite{s12} for cyclic codes and negacyclic codes over $\mathbb{F}_{p^m}+u\mathbb{F}_{p^m}$ of length $4p^s$ can be regarded as corollaries of Corollary \ref{co5.2}, Theorem \ref{th5.3},
Case (III-2) and Case (III-3).

\par
  (vi) For any any nonzero integer $s$, positive integer $n$ satisfying ${\rm gcd}(p,n)=1$.
and $\lambda\in \mathbb{F}_{p^m}^\times$. We know that $(\mathbb{F}_{p^m}+u\mathbb{F}_{p^m})[x]/\langle x^{np^s}-1\rangle$ is not a principal ideal ring for any $s>0$ by Theorems \ref{th2.5} and \ref{th3.8}.
  When $s=0$, $(\mathbb{F}_{p^m}+u\mathbb{F}_{p^m})[x]/\langle x^n-1\rangle$
and $(\mathbb{F}_{p^m}+u\mathbb{F}_{p^m})[x]/\langle x^n+1\rangle$ are principal ideal rings
(cf. \cite{s13} Corollary 3.7, Corollary 5.8 and \cite{s18} Corollary 4.6, respectively).
Hence the study of the paper does not work for $s=0$.

\end{remark}
\section*{Acknowledgments}

\noindent \par
The authors are deeply indebted to the referees and
wish to thank them for their important suggestions and comments. Part of this work was done when Yonglin Cao was visiting Chern Institute of Mathematics, Nankai University, Tianjin, China. Yonglin Cao would like to thank the institution for the kind hospitality.


\medskip
Received xxxx 20xx; revised xxxx 20xx.
\medskip

\end{document}